\documentclass[10pt,letterpaper,english,preprintnumbers,amsmath,amssymb,superscriptaddress,nofootinbib,pra]{revtex4}
\usepackage{mathpazo}
\usepackage[T1]{fontenc}
\usepackage[latin9]{inputenc}
\usepackage{verbatim}
\usepackage{float}
\usepackage{amsthm}
\usepackage{amsmath}
\usepackage{amssymb}
\usepackage{graphicx}

\makeatletter

\pdfpageheight\paperheight
\pdfpagewidth\paperwidth

\@ifundefined{textcolor}{}
{%
 \definecolor{BLACK}{gray}{0}
 \definecolor{WHITE}{gray}{1}
 \definecolor{RED}{rgb}{1,0,0}
 \definecolor{GREEN}{rgb}{0,1,0}
 \definecolor{BLUE}{rgb}{0,0,1}
 \definecolor{CYAN}{cmyk}{1,0,0,0}
 \definecolor{MAGENTA}{cmyk}{0,1,0,0}
 \definecolor{YELLOW}{cmyk}{0,0,1,0}
}
  \theoremstyle{remark}
  \ifx\thechapter\undefined
    \newtheorem{rem}{\protect\remarkname}
  \else
    \newtheorem{rem}{\protect\remarkname}[chapter]
  \fi
  \theoremstyle{plain}
  \ifx\thechapter\undefined
    \newtheorem{lem}{\protect\lemmaname}
  \else
    \newtheorem{lem}{\protect\lemmaname}[chapter]
  \fi
  \theoremstyle{plain}
  \ifx\thechapter\undefined
    \newtheorem{prop}{\protect\propositionname}
  \else
    \newtheorem{prop}{\protect\propositionname}[chapter]
  \fi
  \theoremstyle{plain}
  \ifx\thechapter\undefined
    \newtheorem{cor}{\protect\corollaryname}
  \else
    \newtheorem{cor}{\protect\corollaryname}[chapter]
  \fi
  \theoremstyle{definition}
  \ifx\thechapter\undefined
    \newtheorem{defn}{\protect\definitionname}
  \else
    \newtheorem{defn}{\protect\definitionname}[chapter]
  \fi
  \theoremstyle{plain}
  \newtheorem*{thm*}{\protect\theoremname}
  \theoremstyle{plain}
  \ifx\thechapter\undefined
    \newtheorem{thm}{\protect\theoremname}
  \else
    \newtheorem{thm}{\protect\theoremname}[chapter]
  \fi
  \theoremstyle{plain}
  \newtheorem*{prop*}{\protect\propositionname}

\makeatother

\usepackage{babel}
  \providecommand{\definitionname}{Definition}
  \providecommand{\lemmaname}{Lemma}
  \providecommand{\propositionname}{Proposition}
  \providecommand{\remarkname}{Remark}
  \providecommand{\theoremname}{Theorem}
\providecommand{\corollaryname}{Corollary}
\providecommand{\theoremname}{Theorem}

\begin{document}

\title{The Green's Function for the H{\"u}ckel (Tight Binding) Model}
\author{Ramis Movassagh}
\email{ramis.mov@gmail.com}
\selectlanguage{english}%
\affiliation{Department of Mathematics, IBM T.J. Watson Research Center, Yorktown Heights, NY 10598, USA}

\author{Gilbert Strang}
\affiliation{Department of Mathematics, Massachusetts Institute of Technology,
Cambridge MA, 02139, USA}

\author{Yuta Tsuji}
\affiliation{Education Center for Global Leaders in Molecular Systems for Devices, Kyushu University, Nishi-ku, Fukuoka 819-0395, Japan}

\author{Roald Hoffmann}
\affiliation{Department of Chemistry and Chemical Biology, Cornell University,
Ithaca, NY, 14853, USA}
\begin{abstract}
Applications of the H{\"u}ckel (tight binding) model are ubiquitous
in quantum chemistry and solid state physics. The matrix representation
of this model is isomorphic to an unoriented vertex adjacency matrix
of a bipartite graph, which is also the Laplacian matrix plus twice
the identity. In this paper, we analytically calculate the determinant
and, when it exists, the inverse of this matrix in connection with
the Green's function, $\mathbf{G}$, of the $N\times N$ H{\"u}ckel
matrix. A corollary is a closed form expression for a Harmonic sum
(Eq. \eqref{eq:Math}). We then extend the results to $d-$dimensional
lattices, whose linear size is $N$. The existence of the inverse
becomes a question of number theory. We prove a new theorem in number
theory pertaining to vanishing sums of cosines and use it to prove
that the inverse exists if and only if $N+1$ and $d$ are odd and
$d$ is smaller than the smallest divisor of $N+1$. We corroborate
our results by demonstrating the entry patterns of the Green's function
and discuss applications related to  transport and conductivity.
\end{abstract}
\maketitle

\section{Problem Statement: The H{\"u}ckel model and its Green's Function}
The H{\"u}ckel or Tight Binding model was originally introduced to
describe electron hopping on a one-dimensional chain or ring \cite{Heilbronner1976}.
It has come to serve as a ubiquitous model in solid state chemistry
and physics \cite{Heilbronner1976,ashcroft1981solid}. Two typical
forms of the H{\"u}ckel matrix, for a linear chain of $N$ atoms,
and for a cycle of $N$ atoms, are given in Eq. \eqref{eq:Cycle6Atoms}.
The resulting banded matrix is isomorphic to the vertex adjacency
matrix of a graph \cite{Gutman1986}. 

The diagonal entries of the H{\"u}ckel Hamiltonian matrix are defined by the Coulomb integral, $\alpha=\langle\chi_i|\hat{H}|\chi_i\rangle$, where $\chi_i$ is the basis function of $2p_{z} AO$ of the $i^{th}$ carbon atom. The magnitude of $\alpha$ can be approximated by the ionization potential of a carbon atom. It is reasonable to assume that all carbon atoms have the same ionization potential, resulting in an approximation that all diagonal elements have the same value, which is denoted by just $\alpha$. 

 An off-diagonal elements of the H{\"u}ckel Hamiltonian matrix is typically called a resonance integral, $\beta_{i,j}$, defined as $\beta_{i,j}=\langle\chi_i|\hat{H}|\chi_j\rangle$  . This is a measure of the interaction between the $i^{th}$ and $j^{th}$ carbon atoms. Usually $\beta_{i,j}$ is neglected if there is no bond between the $i^{th}$ and $j^{th}$ carbon atoms. It is reasonably assumed that all the carbon-carbon bonds have the same strength, resulting in an approximation that all non-zero off-diagonal elements have the same value, which is denoted by just $\beta$.

If there are no heteroatoms in the system, all Coulomb integrals can be considered identical. However, if there is a heteroatom, we need to modify the Coulomb integral properly. When one calculates a linear $\pi -$conjugated chain, namely  a polyene, all resonance integrals cannot be considered identical, if there is bond alternation, as obtains normally. We address this problem later on. In a small $\pi -$conjugated cycle, termed annulene, there is no bond alternation, so all resonance integrals can be considered identical. However, as the cycle gets large, bond alternation sets in (see refs. \cite{longuet1960alternation, wannere2004aromaticity}).
\subsection{Mathematical Statement of the Problem}
Therefore, in the simplest version of both physical and chemical models, the
matrix representation of the electronic Hamiltonian for a network
of orbitals (one per atom, say carbon $2p_{z}$), is characterized
by diagonal matrix elements $\alpha$, which we may without loss of
generality set equal to zero, and off-diagonal elements $\beta$ (or
$t$), where two atoms are neighbors. If we use units of $\beta$
($\beta$ is negative), we may replace these nearest neighbor interactions
by unity, $1$. All other (non-nearest neighbor) interactions are
set equal to $0$.

The eigenvalues and eigenvectors of these matrices are well known
for the linear chain and cycles \cite{Heilbronner1976}. For the finite
linear chain (the banded matrix at left in Eq. \eqref{eq:Cycle6Atoms}),
they can be written respectively as cosines and sines, as follows.
Let $\lambda_{r}$ and $|\psi_{r}\rangle$ be the $r^{\mbox{th}}$
eigenvalue and the corresponding $r^{\mbox{th}}$ eigenvector, then
\begin{eqnarray}
\lambda_{r} & = & 2\cos r\omega\qquad\mbox{in units of }\beta\label{eq:Eval}\\
|\psi_{r}\rangle & = & \sqrt{\frac{2}{\left(N+1\right)}}\left[\sin\left(r\omega\right),\sin\left(2r\omega\right),\dots,\sin\left(Nr\omega\right)\right]^{T},\label{eq:Evec}
\end{eqnarray}
where $\omega\equiv\frac{\pi}{N+1}$ and $1\le r\le N$ is an integer.
The H{\"u}ckel matrix for a linear chain, whose
matrix representation is a symmetric tridiagonal matrix \cite[for a review]{meurant1992review},
and a $N-$membered ring, which is a circulant
matrix \cite{davis1979circulant,gray2005toeplitz}, respectively are
\begin{equation}
H_{1}=\left[\begin{array}{cccccc}
0 & 1\\
1 & 0 & 1\\
 & 1 & 0 & 1\\
 &  & 1 & 0 & \ddots\\
 &  &  & \ddots & \ddots & 1\\
 &  &  &  & 1 & 0
\end{array}\right]\qquad H_{1}^{c}=\left[\begin{array}{cccccc}
0 & 1 &  &  &  & 1\\
1 & 0 & 1\\
 & 1 & 0 & 1\\
 &  & 1 & 0 & \ddots\\
 &  &  & \ddots & \ddots & 1\\
1 &  &  &  & 1 & 0
\end{array}\right],\label{eq:Cycle6Atoms}
\end{equation}
where to emphasize the one-dimensional structure of the molecular system we put a subscript $1$ on $H$ and denote the Hamiltonian by $H_1$. From now omitted entries are zeros; we will explicitly write
down the zeros when it helps the presentation.

The H{\"u}ckel model has found renewed significance in recent experimental
and theoretical studies of molecular conductance, that is transmission
of a current through a molecule \cite[and references therein]{Solomon2013}. 

The Green's function matrix, $\mathbf{G}$, is defined via the resolvent
as
\begin{equation}
G\left(r,s\mbox{ };\mbox{ }E\right)=\langle r|\frac{1}{E-H}|s\rangle,\label{eq:GreensFunction_def}
\end{equation}
where $G\left(r,s\mbox{ };\mbox{ }E\right)$ is the $r,s$ entry of
the Green's function matrix, $H$ is the Hamiltonian and $E$ is an
energy. 

The Green's function plays an important role in the calculation of
transport phenomena such as conductivity \cite{Datta2005}. In the
simplest form of the theory, the conductance between electrodes connected
to sites $r$ and $s$ of a molecule is proportional to the square
of the absolute value of the matrix element of the unperturbed Green\textquoteright s
function,
\begin{equation}
G^{\left(0\right)}\left(r,s\mbox{ };\mbox{ }E\right)=\sum_{k}\frac{C_{rk}C_{sk}^{*}}{E-\epsilon_{k}+i\eta},\label{eq:Grs}
\end{equation}
where $C_{rk}$ is the coefficient of the $r^{\mbox{th}}$ atomic
orbital in the $k^{\mbox{th}}$ molecular orbital (MO) in an orthogonal
basis, $\epsilon_{k}$ is the $k^{\mbox{th}}$ MO energy, and $\eta$
is an infinitesimal positive number to assure analyticity. The Fermi
energy is equal to the Coulomb integral of the H{\"u}ckel model and
for convenience we set this energy to zero (see subsection "Level set of the Fermi Energy" below).

For the DC conductivity we want to evaluate the Green's function at
the Fermi energy; therefore $E=0$ in Eq. \eqref{eq:GreensFunction_def}.
By Sokhotski\textendash Plemelj theorem, $\lim_{\eta\rightarrow 0^+}\frac{1}{x\pm\mathrm{i}\eta} = {P}(\frac{1}{x})\mp i\pi\delta(x)$,
the real part of the Green's function (Eq. \eqref{eq:Grs}) is 
\begin{equation}
G^{\left(0\right)}\left(r,s\mbox{ };\mbox{ }E=0\right)=-\sum_{k}\frac{C_{rk}C_{sk}^{*}}{\epsilon_{k}}\quad.\label{eq:GrsApprox}
\end{equation}
Therefore, the Green's function in the basis described above (Eqs.
\ref{eq:Eval} and \ref{eq:Evec}), for a finite open linear chain
has entries that are:

\begin{eqnarray}
G^{\left(0\right)}\left(r,s\mbox{ };\mbox{ }E=0\right) & = & -\frac{1}{N+1}\sum_{k=1}^{N}\frac{\sin\left(rk\omega\right)\sin\left(sk\omega\right)}{\cos\left(k\omega\right)}\qquad\mbox{in units of }\beta^{-1}\label{eq:Problem}\\
\omega & \equiv & \frac{\pi}{N+1}\quad.\nonumber 
\end{eqnarray}
Below, for simplicity, we denote $G\left(r,s\right)\equiv G^{\left(0\right)}\left(r,s\mbox{ };\mbox{ }E=0\right)$. \begin{rem}
$\mathbf{G}^{(0)}$ is simply $-H_{1}^{-1}$ in the basis given by
Eqs. \eqref{eq:Eval} and \eqref{eq:Evec}. Generally, with energy set-point
$E=0$, the Green's function is minus the inverse of the Hamiltonian.
We compute the inverses of $H_{1}$ and $H_{1}^{c}$ in several ways--
from general formulas for tridiagonal matrices, directly from simple
equations for the first column of the inverses, and from factoring
the matrix symbol $e^{i\theta}+e^{-i\theta}$. 
\end{rem}
Our goal is to prove the conditions under which the inverse of various forms of the H{\"u}ckel model exists for different $N$ and in $d-$spatial dimensions. When it exists, we analytically derive closed-form formulas for the Green's function $\mathbf{G}^{(0)}$.

\subsection{Level set of the Fermi Energy}
The position of the actual Fermi levels in a calculation of molecular transmission may vary. It has proven to be a good approximation to set it equal to the Coulomb integral of the H{\"u}ckel method for most rings and chains (the limitations of this assumption will be mentioned later).

There is a good reason why we assume $E_F=\alpha$. The energy level of the $2p$ atomic orbital (AO) of carbon ($-11.4 eV$, the same energy level as the Coulomb integral), is almost the same as that of the $6s$ AO of the widely used Au electrode $(-10.9 eV)$ \cite{alvarez1989tables}. Note that the electronic configuration of Au is $5d^{10}6s^1$. This approximation works well as long as significant charge transfer between the molecule and the electrode surface does not occur \cite{xue2002first}.

The assumption that the Fermi level is equal to the Coulomb integral of the H{\"u}ckel  method is probably valid for even-membered chains and rings with $4n+2$ atoms.
\section{Determinants and Analytical Expressions for $H_{1}^{-1}$ and $\left(H_{1}^{c}\right)^{-1}$}
\subsection{Open Chain, $H_{1}$}
\begin{lem}
$H_{1}$ is only invertible when $N$ is even, in which case $\mbox{det}\left(H_{1}\right)=\left(-1\right)^{N/2}$\end{lem}
\begin{proof}
Generally for any $N$ 
\begin{eqnarray*}
\det\left(H_{N}\right) & = & -\det\left(H_{N-2}\right)=\cdots=\left(-1\right)^{\frac{N-2}{2}}\det\left(H_{2}\right)=\left(-1\right)^{\frac{N-2}{2}}\left|\begin{array}{cc}
0 & 1\\
1 & 0
\end{array}\right|=\left(-1\right)^{N/2}\quad N\mbox{ even,}\\
\det\left(H_{N}\right) & = & -\det\left(H_{N-2}\right)=\cdots=\det\left(H_{1}\right)=\left|0\right|=0\quad\quad N\mbox{ odd},
\end{eqnarray*}
where we denoted $H_{1}$ of size $N\times N$, simply by $H_{N}$
and the determinant of a matrix by $\left|\cdot\right|$. In Eq. \eqref{eq:Problem},
$\cos\left(k\omega\right)$ can take on a zero value if $N$ is odd,
whereby $G\left(r,s\right)$ is not defined.\end{proof}
\begin{prop}
When $N$ is even, the entries of the Green's function $G\left(r,s\right)\equiv-H_{1}^{-1}\left(r,s\right)$
are
\begin{equation}
G\left(r,s\right)=\left\{ \begin{array}{ccc}
\left(-1\right)^{\frac{r+s-1}{2}} & \quad & r<s\mbox{ }:\mbox{ }r\mbox{ odd and }s\mbox{ even}\\
\left(-1\right)^{\frac{r+s-1}{2}} & \quad & r>s\mbox{ }:\mbox{ }r\mbox{ even and }s\mbox{ odd}\\
0 & \quad & \mbox{otherwise}.
\end{array}\right.\label{eq:G_rs_Analytical_Final}
\end{equation}
\begin{equation}
G=-H_{1}^{-1}=\left[\begin{array}{ccccccc}
0 & -1 & 0 & +1 & 0 & -1 & \cdots\\
-1 & 0 & 0 & 0 & 0 & 0\\
0 & 0 & 0 & -1 & 0 & +1\\
+1 & 0 & -1 & 0 & 0 & 0 & \cdots\\
0 & 0 & 0 & 0 & 0 & -1\\
-1 & 0 & +1 & 0 & -1 & 0 & \cdots\\
\vdots &  &  & \vdots &  &  & \ddots
\end{array}\right]\label{eq:The-Green's-function}
\end{equation}
\end{prop}
\begin{proof}
Suppose we have a general tridiagonal matrix
\[
A=\left[\begin{array}{ccccc}
b_{1} & c_{1}\\
a_{1} & b_{2} & c_{2}\\
 & a_{2} & b_{3} & \ddots\\
 &  & \ddots & \ddots & c_{N-1}\\
 &  &  & a_{N-1} & b_{N}
\end{array}\right]
\]
Then Usmani's formula \cite{usmani1994inversion} for the $r,s$ entry
of $A^{-1}$ is 
\begin{equation}
\alpha\left(r,s\right)=\left\{ \begin{array}{ccc}
\left(-1\right)^{r+s}c_{r}c_{r+1}\cdots c_{s-1}\mbox{ }\mbox{ }\theta_{r-1}\phi_{s+1}/\theta_{N} & \quad & r<s\\
\theta_{r-1}\phi_{r+1}/\theta_{N} & \quad & r=s\\
\left(-1\right)^{r+s}a_{s+1}a_{s+2}\cdots a_{r}\mbox{ }\mbox{ }\theta_{s-1}\phi_{r+1}/\theta_{N} & \quad & r>s
\end{array}\right.\label{eq:General_TriInverse}
\end{equation}
where $\theta_{r}$ and $\phi_{s}$ satisfy second order recursion
relations
\[
\begin{array}{ccc}
\theta_{r}=b_{r}\theta_{r-1}-a_{r}c_{r-1}\theta_{r-2} & \quad & r=1,2,\dots,N-1,N\\
\phi_{s}=b_{s}\phi_{s+1}-c_{s}a_{s+1}\phi_{s+2} & \quad & s=N,N-1,\cdots2,1
\end{array}
\]
with the initial conditions $\theta_{-1}=0$, $\theta_{0}=1$, $\phi_{N+1}=1$
and $\phi_{N+2}=0$. 

We are interested in the special case where $b_{i}=0$, $c_{i}=a_{i}=1$
for all $i$. The recursion relations are now given by
\begin{eqnarray*}
\theta_{r} & = & -\theta_{r-2}\qquad r=2,3,\dots,N\\
\phi_{s} & = & -\phi_{s+2}\qquad s=N,N-1,\dots,2,1
\end{eqnarray*}
The solutions, after imposing the initial conditions, are 
\begin{eqnarray*}
\theta_{r} & = & \frac{i^{r}}{2}\left[1+\left(-1\right)^{r}\right];\qquad\phi_{s}=\frac{i^{-\left(N+1\right)+s}}{2}\left[1-\left(-1\right)^{s}\right].
\end{eqnarray*}
Substituting these in Eq. \eqref{eq:General_TriInverse} and multiplying
by $-1$, we obtain $\mathbf{G}=-H_{1}^{-1}$: 
\begin{equation}
G\left(r,s\right)=\left\{ \begin{array}{ccc}
\frac{i^{3\left(r+s-1\right)}}{4}\left[1+\left(-1\right)^{r-1}\right]\left[1-\left(-1\right)^{s+1}\right] & \quad & r<s\\
\mbox{ }\frac{i^{i^{2\left(r-1\right)}}}{4}\left[1+\left(-1\right)^{r-1}\right]\left[1-\left(-1\right)^{r+1}\right] & \quad & r=s\\
\frac{i^{3\left(r+s-1\right)}}{4}\left[1+\left(-1\right)^{s-1}\right]\left[1-\left(-1\right)^{r+1}\right] & \quad & r>s
\end{array}\right.\label{eq:G_rs_Analytical}
\end{equation}
\end{proof}
By symmetry we may focus on $r\ge s$. In Eq. \eqref{eq:G_rs_Analytical}
the only nonzero elements, for $r\ge s$, correspond to $r$ even
and $s$ odd, in which case $G\left(r,s\right)=i^{3\left(r+s-1\right)}=\left(-1\right)^{\frac{r+s-1}{2}}$. 
\begin{cor}
The closed form expression for the following sum gives the identity
\begin{equation}
-\frac{1}{N+1}\sum_{k=1}^{N}\frac{\sin\left(rk\omega\right)\sin\left(sk\omega\right)}{\cos\left(k\omega\right)}=\left(-1\right)^{\frac{r+s-1}{2}}\label{eq:Math}
\end{equation}
when $r$ is even and $s<r$ is odd. Otherwise, $G\left(r,s\right)=0$
when $s\le r$ and $G\left(r,s\right)=G\left(s,r\right)$ when $r\le s$. 
\end{cor}
This is the Green's function in an orthonormal basis. A purely trigonometric
derivation, that does not use the H{\"u}ckel matrix and serves as
an alternative proof of Eqs. \eqref{eq:G_rs_Analytical_Final} and \eqref{eq:Math},
is presented in the appendix.
\begin{rem}
Formulas for the inverse of a tridiagonal Toeplitz matrix have been
given by Schlegel \cite{schlegel1970explicit} and Mallik \cite{mallik2001inverse}
in terms of Chebyshev polynomials. 
\end{rem}
In quantum chemistry, it was known that $G\left(r,s\right)=0$ when
$r$ and $s$ have the same parity. These zeros can be derived from
a property called \textquotedblleft alternancy\textquotedblright{}
(the original proof is due to C. A. Coulson and G. S. Rushbrooke \cite{coulson1940note}).
If the interacting orbitals of a molecule can be divided into two
disjoint sets, where the atoms of one set are adjacent only to atoms
of the other set, the molecule is said to be alternant. For alternants,
for instance the linear chain studied here, a number of results can
be proved; for instance the energy levels are paired positive and
negative, and in paired levels the coefficients of one set of atoms
are just minus the coefficients of that set in the paired level. It
follows that $G\left(r,s\right)=0$ when $r$ and $s$ have the same
parity. The other zeros and $\pm1$ entries, as far as we know, were
not noticed.

In chemical applications one often has to deal with the special case
of alternating bond strengths along a chain. The proposition below
gives the form of the Hamiltonian and its corresponding Green's function.
\begin{defn}
The bond alternating Hamiltonian is defined by $H_{alt}$ for $N$
even, where 
\[
H_{alt}=\left[\begin{array}{ccccccc}
0 & \beta\\
\beta & 0 & \alpha\\
 & \alpha & 0 & \beta\\
 &  & \beta & 0 & \ddots\\
 &  &  & \ddots & \ddots & \alpha\\
 &  &  &  & \alpha & 0 & \beta\\
 &  &  &  &  & \beta & 0
\end{array}\right]
\]
\end{defn}
Comment: In this special limit, the Toeplitz structure is lost. $\alpha$
in this definition is not related to the one discussed above, which
stood for the diagonal elements and was taken to be zero.
\begin{prop}
\label{prop:BondAltern}The entries of the Green's function $\mathbf{G}\equiv-\left(H_{alt}\right)^{-1}$,
are given by 
\begin{eqnarray}
G\left(r,s\right) & = & \left\{ \begin{array}{ccc}
\frac{\left(-1\right)^{\frac{r+s-1}{2}}}{4}\frac{1}{\beta}\left(\frac{\alpha}{\beta}\right)^{\frac{r-s-1}{2}}\left\{ 1-\left(-1\right)^{s}\right\} \left\{ 1+\left(-1\right)^{r}\right\} , & \qquad & r\ge s\\
\\
\frac{\left(-1\right)^{\frac{r+s-1}{2}}}{4}\frac{1}{\beta}\left(\frac{\alpha}{\beta}\right)^{\frac{s-r-1}{2}}\left\{ 1-\left(-1\right)^{r}\right\} \left\{ 1+\left(-1\right)^{s}\right\} , &  & r\le s\mbox{ }.
\end{array}\right.\label{eq:G_bondAlternation0}
\end{eqnarray}
\end{prop}
\begin{proof}
The form can be derived using the same techniques as above.
\end{proof}
\subsection{Cyclic chain, $H_{1}^{c}$}
The cyclic Hamiltonian is a circulant matrix and therefore diagonalizable
in Fourier basis \cite{davis1979circulant,gray2005toeplitz,strang1999discrete}.
Let $\omega_{j}=\exp\left(2\pi ij/n\right)$, the eigenpairs are
\begin{eqnarray*}
\lambda_{j} & = & 2\cos\left(2\pi j/n\right),\qquad j=0,1,\dots,n-1\\
v_{j}^{T} & = & \frac{1}{\sqrt{n}}\left(1,\omega_{j},\omega_{j}^{2},\dots,\omega_{j}^{\left(n-1\right)}\right).
\end{eqnarray*}
In particular when $n=4j$, the matrix has zero eigenvalues and hence
non-invertible. The following lemma sharpens this notion.
\begin{lem}
\label{lem:Cycle-determinat}The determinant of $H_{1}^{c}$ is given
by
\[
\det\left(H_{1}^{c}\right)=\left\{ \begin{array}{ccc}
-1 & \quad & N=2\\
2 & \quad & N=2k+1\\
0 & \quad & N=4k\\
-4 & \quad & N=4k+2
\end{array}\right.,\qquad for\quad k\in\mathbb{N}.
\]
\end{lem}
\begin{proof}
When $N=2$, trivially $\left|\begin{array}{cc}
0 & 1\\
1 & 0
\end{array}\right|=-1$. When $N$ is odd, we express
\[
H_{1}^{c}=\left[\begin{array}{cc}
H_{N-1} & B\\
C & D
\end{array}\right]
\]
where $C=\left[\begin{array}{ccccc}
1 & 0 & \cdots & 0 & 1\end{array}\right]$, $B=C^{T}$, $D=0$ and $H_{N-1}$ is an $N-1\times N-1$ version
of $H_{1}$ (open chain) as defined above, which is invertible. With
this decomposition, the structure of $H_{1}^{-1}$ derived above,
and the well-known fact about the determinant of block matrices we
arrive at
\[
\det\left(H_{1}^{c}\right)=-\det\left(H_{N-1}\right)\det\left(CH_{N-1}^{-1}B\right)=\det\left(CH_{N-1}^{-1}B\right)=2\quad.
\]
When $N$ is a multiple of $4$, one can easily check that the vectors
$\mathbf{v}_{1}\equiv\left[0,-1,0,1\right]^{T}$ and $\mathbf{v}_{2}\equiv\left[1,0,-1,0\right]^{T}$
generate the kernel of $H_{1}^{c}$. Namely, if $H_{1}^{c}$ is a
$4k\times4k$ matrix, then the $k-$fold concatenations $[\mathbf{v}_{1}\mathbf{v}_{1}\cdots\mathbf{v}_{1}]^{T}$
and $\left[\mathbf{v}_{2}\mathbf{v}_{2}\cdots\mathbf{v}_{2}\right]^{T}$
are in the $\ker\left(H_{1}^{c}\right)$. Moreover, since excluding
the last two rows and columns of $H_{1}^{c}$ gives $H_{4k-2}$, which
is invertible, we conclude that the two vectors are a basis for the
kernel of $H_{1}^{c}$. 

Lastly, if $N$ is even yet not a multiple of $4$, we write $H_{1}^{c}=\left[\begin{array}{cc}
H_{N-2} & B\\
B^{T} & D
\end{array}\right]$, where $B^{T}=\left[\begin{array}{cccc}
0 & \cdots & 0 & 1\\
1 & 0 & \cdots & 0
\end{array}\right]$ and $D=\left[\begin{array}{cc}
0 & 1\\
1 & 0
\end{array}\right]$. Here $H_{N-2}$ has a size that is a multiple of $4$. Using the
techniques above we obtain
\begin{eqnarray*}
\det\left(H_{1}^{c}\right) & = & \det\left(H_{N-2}\right)\det\left(D-B^{T}H_{N-2}^{-1}B\right)=\det\left(D-B^{T}H_{N-2}^{-1}B\right)\\
 & = & \det\left(\left[\begin{array}{cc}
0 & 1\\
1 & 0
\end{array}\right]-\left[\begin{array}{cc}
0 & -1\\
-1 & 0
\end{array}\right]\right)=-4\quad.
\end{eqnarray*}
\end{proof}
\begin{defn}
A Toeplitz matrix is a matrix that is constant along diagonals. A
circulant matrix is Toeplitz, and each column is a cyclic shift of
the previous column \cite{TrefethenEmbree2005,gray2005toeplitz}.
Thus the lower triangular part of a circulant determines the upper
triangular part:
\end{defn}
\[
\mbox{Toeplitz }A=\left[\begin{array}{ccc}
x_{0} & x_{-1} & x_{-2}\\
x_{1} & x_{0} & x_{-1}\\
x_{2} & x_{1} & x_{0}
\end{array}\right]\quad\mbox{or}\quad\left[\begin{array}{cccc}
x_{0} & x_{-1} & \cdots & x_{-\left(N-1\right)}\\
x_{1} & x_{0} &  & \vdots\\
\vdots &  & \ddots & x_{-1}\\
x_{N-1} & \cdots & x_{1} & x_{0}
\end{array}\right]
\]

\[
\mbox{Circulant }A=\left[\begin{array}{ccc}
x_{0} & x_{2} & x_{1}\\
x_{1} & x_{0} & x_{2}\\
x_{2} & x_{1} & x_{0}
\end{array}\right]\quad\mbox{or}\quad\left[\begin{array}{cccc}
x_{0} & x_{N-1} & \cdots & x_{1}\\
x_{1} & x_{0} &  & x_{2}\\
\vdots &  & \ddots & \vdots\\
x_{N-1} & x_{N-2} & \cdots & x_{0}
\end{array}\right].
\]
The inverse of a circulant matrix is circulant. The inverse of a Toeplitz
matrix is not in general Toeplitz.

A Toeplitz matrix $T$ has $\left(r,s\right)$ entries that depend
on $r-s$. Therefore specifying the first row and the first column
fully specifies the matrix. Specifying the first column(or row) is
sufficient to specify a circulant matrix. 
\begin{prop}
\label{prop:Cycle-Green's-function}The $N\times N$ H{\"u}ckel circulant
matrix $H_{1}^{c}$ is invertible for $N\ne4k$. The first column
of the inverse is 
\[
\begin{array}{ccccc}
N=4k+1 & \quad & \left(x_{0},x_{1},x_{2},x_{3},\dots\right) & = & \frac{1}{2}\left(1,1,-1,-1,\mbox{ repeat}\right)\\
N=4k+2 & \quad & \left(x_{0},x_{1},x_{2},x_{3},\dots\right) & = & \frac{1}{2}\left(0,1,0,-1,\mbox{ repeat}\right)\\
N=4k+3 & \quad & \left(x_{0},x_{1},x_{2},x_{3},\dots\right) & = & \frac{1}{2}\left(-1,1,1,-1,\mbox{ repeat}\right)
\end{array}
\]
The matrix representation is shown in Fig. \ref{fig:Toeplitz-structure-of}.
\end{prop}
\begin{figure}[H]
\raggedright{}\includegraphics[scale=0.22]{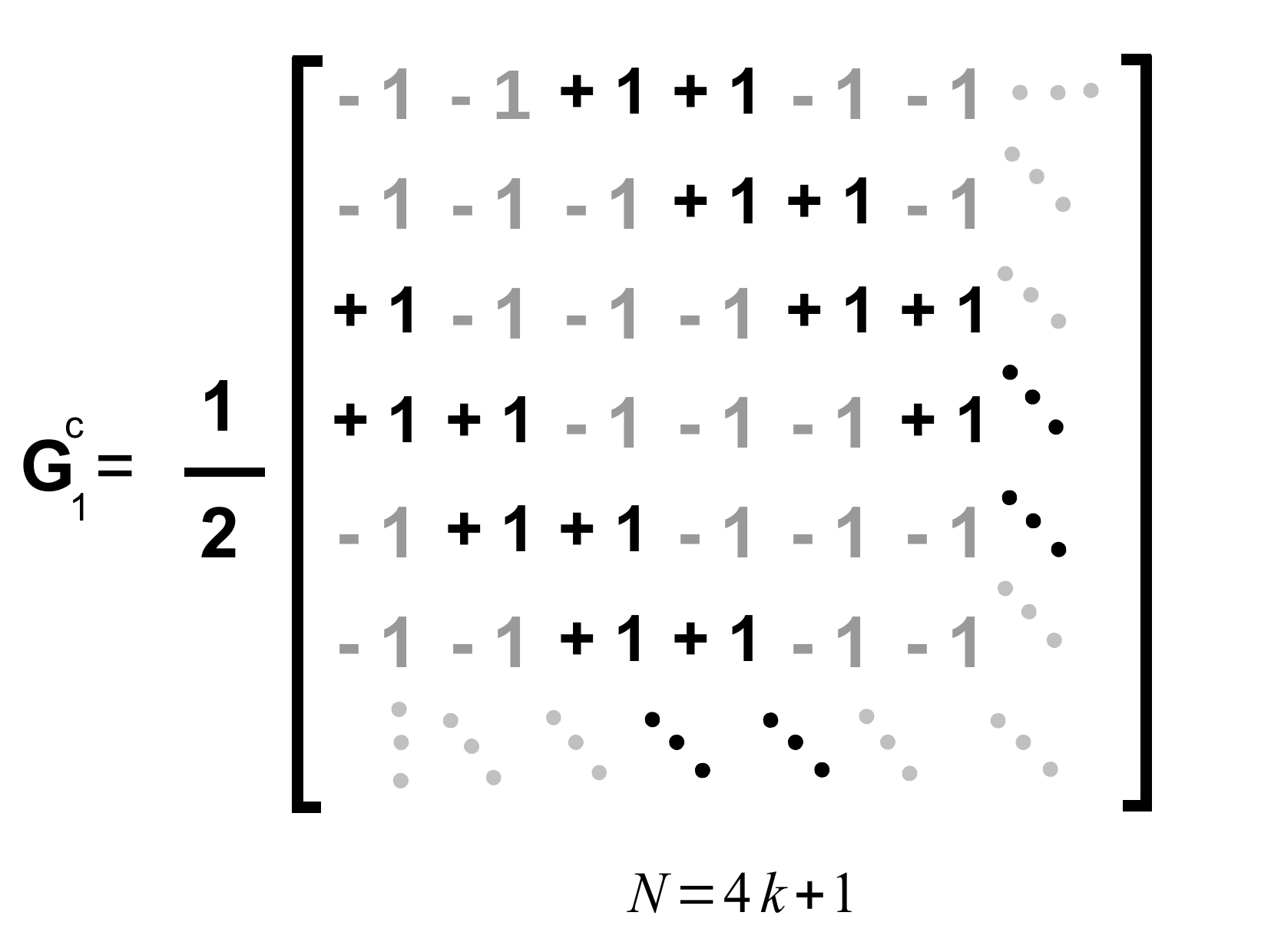}$\quad$\includegraphics[scale=0.22]{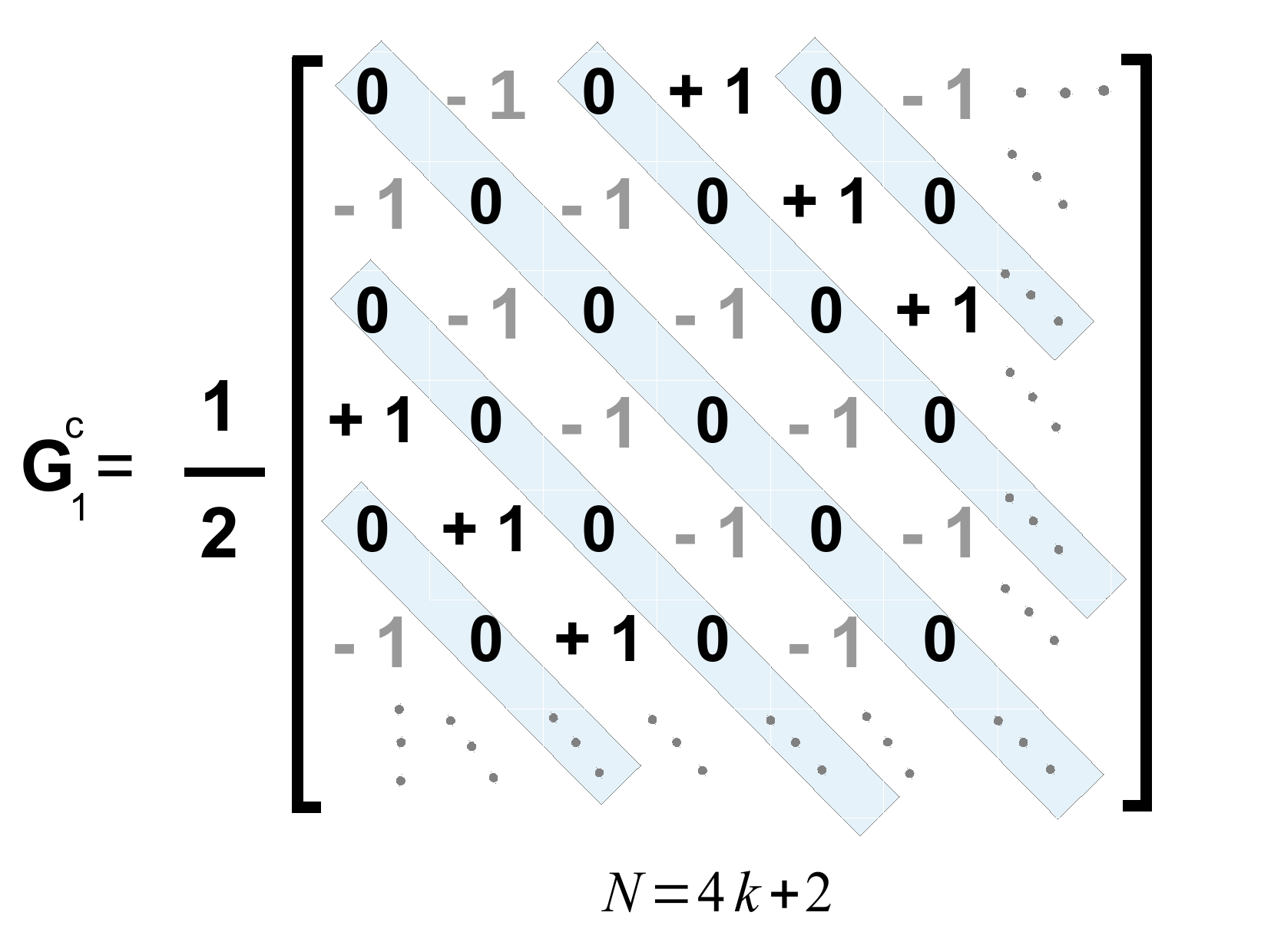}$\quad$\includegraphics[scale=0.22]{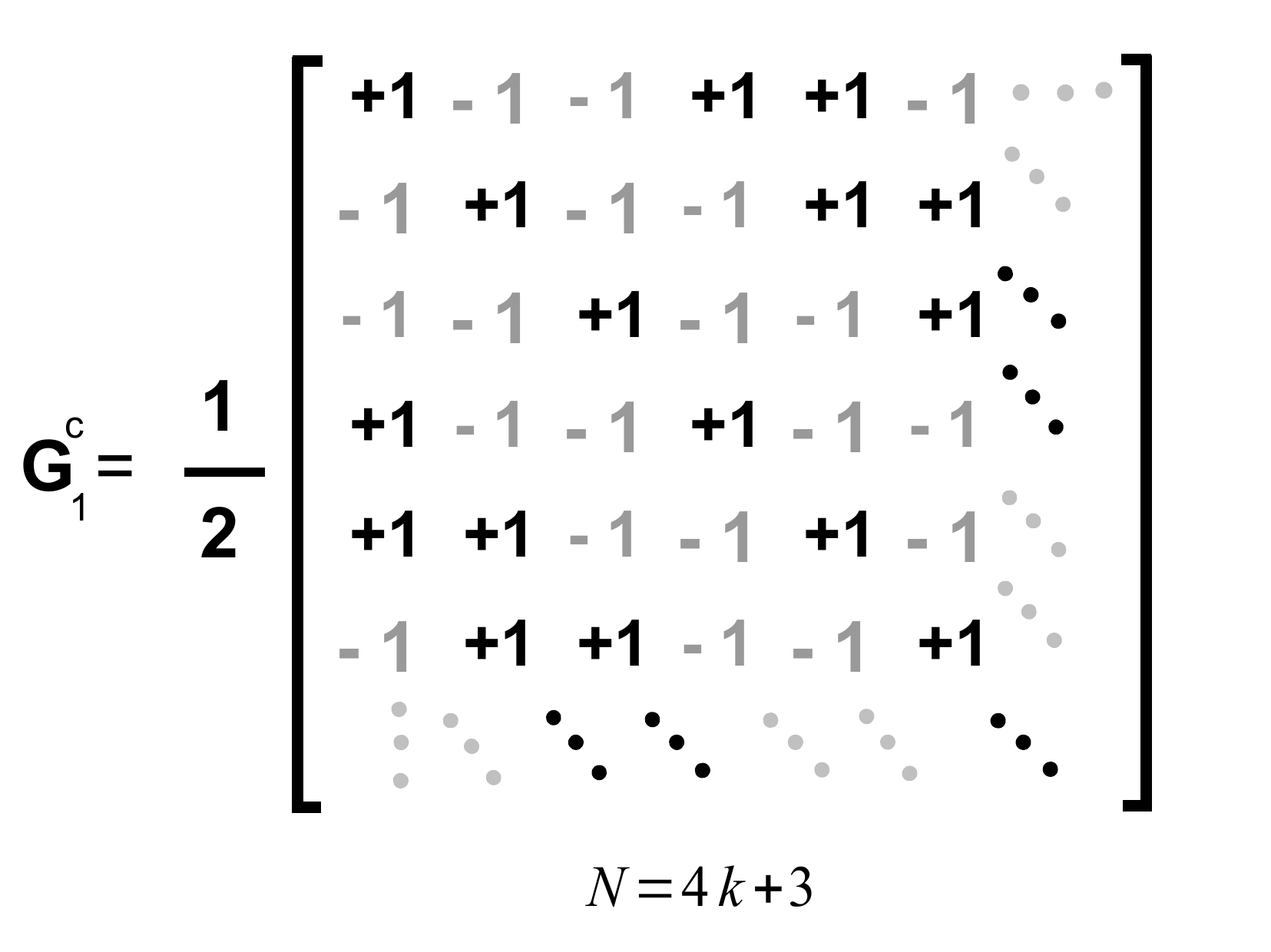}\protect\caption{\label{fig:Toeplitz-structure-of}The Toeplitz structure of $\mathbf{G}_{1}^{c}=-\left(H_{1}^{c}\right)^{-1}$. }
\end{figure}
\begin{proof}
$H_{1}^{c}$ (see Eq. \eqref{eq:Cycle6Atoms}), is a symmetric circulant
matrix, so its inverse is also a symmetric circulant. Thus $ x_k = x_{N-k} $
for $ 0 < k < N/2 $. The matrix $ H^c_1 $ (with two cyclic diagonals
of 1's) multiplies the first column $ (x_0, \ldots, x_{N-1}) $ of
its inverse to give the first column of the identity matrix: 
\[
\left[\begin{array}{ccccc}
0 & 1 &  &  & 1\\
1 & 0 & 1\\
 & 1 & 0 & \ddots\\
 &  & \ddots & \ddots & 1\\
1 &  &  & 1 & 0
\end{array}\right]\left[\begin{array}{c}
x_{0}\\
x_{1}\\
x_{2}\\
\vdots\\
x_{N-1}
\end{array}\right]=\left[\begin{array}{c}
1\\
0\\
0\\
\vdots\\
0
\end{array}\right]
\]

In the first row, symmetry changes $ x_1 + x_{N-1} = 1 $ to $ 2x_1 = 1 $
and $ x_1 = \frac{1}{2}$. Then the odd-numbered rows produce $ x_3, x_5, x_7, \ldots $
with alternating signs:
\begin{align}
x_{1}+x_{3} & =0\quad\mbox{so that }x_{3}=-\frac{1}{2}\nonumber \\
x_{3}+x_{5} & =0\quad\mbox{so that }x_{5}=+\frac{1}{2}\label{eq:Gil1}\\
\mbox{etc.}\nonumber 
\end{align}

The even numbered rows also produce alternating signs:
\begin{align}
x_{0}+x_{2} & =0\quad\mbox{so that }x_{2}=-x_{0}\label{eq:Gil2}\\
x_{2}+x_{4} & =0\quad\mbox{so that }x_{4}=+x_{0}\nonumber \\
\mbox{etc.}\nonumber 
\end{align}

Finally the last row gives $x_0 = -x_{N-2}$. This produces the three
separate possibilities for the inverse matrix in Proposition 2:
\[
\begin{array}{ccc}
N=4k+1 & \quad & x_{N-2}=x_{4k-1}=-\frac{1}{2}\mbox{ by Eq. \eqref{eq:Gil1} and then }x_{0}=\frac{1}{2}\\
N=4K+2 & \quad & x_{N-2}=x_{4k}=x_{0}\mbox{ by Eq. \eqref{eq:Gil2} and then }x_{0}=-x_{0}=0\\
N=4k+3 & \quad & x_{N-2}=x_{4k+1}=+\frac{1}{2}\mbox{ by Eq. \eqref{eq:Gil1} and then }x_{0}=-\frac{1}{2}.
\end{array}
\]
The alternating signs for $x_{0},x_{2},x_{4},\dots$ complete the
inverse circulant matrix $\left(H_{1}^{c}\right)^{-1}$. The Green's
matrix is defined as $\mathbf{G}_{1}^{c}\equiv-\left(H_{1}^{c}\right)^{-1}$
and is shown in Fig. \ref{fig:Toeplitz-structure-of}.\end{proof}
\begin{rem}
The same direct approach produces $H_{1}^{-1}$ in the non-circulant
case. The $(1,n)$ and $(n,1)$ entries of $H_{1}^{c}$ are now set
to zero. Then the first equation in Eq. \eqref{eq:Gil1} is simply $x_{1}=1$.
The other equations in Eq. \eqref{eq:Gil1} give alternating signs for
$x_{3},x_{5},\dots$ . 

Similarly, the last equation in Eq. \eqref{eq:Gil2} is now $x_{N-2}=0$.
The first column $(x_{0},x_{1},\dots)$ of $H_{1}^{-1}$ is seen to
be $(0,1,0,-1,\mbox{{\it repeat}})$. The last column of $H_{1}^{-1}$
has these components in reverse order. Then by symmetry we also know
the first and last rows of $H_{1}^{-1}$. 

Because $H_{1}$ is tridiagonal, these two columns and two rows completely
determine the rest of $H_{1}^{-1}$. On and above the main diagonal,
all sub-matrices of $H_{1}^{-1}$ have rank $1$. (If $H_{1}$ is
tridiagonal and invertible then $H_{1}^{-1}$ is a ``semi-separable''
matrix \cite{vandebril2007matrix}.) It is easy to see that starting
from the first and last rows and columns of $\mathbf{G}$ in Proposition
\ref{prop:BondAltern}, all other entries of $\mathbf{G}$ follow
directly from the rank $1$ requirement.\end{rem}
\begin{proof}
(alternative to Prop. \ref{prop:Cycle-Green's-function}) We now establish
the inverse of $H_{1}^{c}$ using a technique that is general to circulant
matrices based on the factorization of the symbol. $H_{1}^{c}=S+S^{-1}=S+S^{T}$,
where $S$ is the $N\times N$ cyclic shift matrix: $S^{N}=I$. Since
$S+S^{-1}=\left(\mathbb{I}-iS\right)\left(\mathbb{I}+iS\right)S^{-1}$,
we have 
\begin{eqnarray}
\left(S+S^{-1}\right)^{-1} & = & S\left(\mathbb{I}+iS\right)^{-1}\left(\mathbb{I}-iS\right)^{-1}\nonumber \\
 & = & S\frac{\left(\mathbb{I}+\left(-i\right)S+\cdots+\left(-i\right)^{N-1}S^{N-1}\right)}{1-\left(-i\right)^{N}}\frac{\left(\mathbb{I}+iS+\cdots+i^{N-1}S^{N-1}\right)}{1-i^{N}}\quad.\label{eq:GilAlternative}
\end{eqnarray}
The denominator is $\left(1-i^{N}\right)\left(1-\left(-i\right)^{N}\right)=\left\{ \begin{array}{c}
0\quad N\equiv4k\quad\\
2\quad N\equiv4k+1\\
4\quad N\equiv4k+2\\
2\quad N\equiv4k+3
\end{array}\right.$. This confirms that $H_{1}^{c}$ is singular for $N=4k$. 

The coefficient of $S^{N}$ in the product given by Eq. \eqref{eq:GilAlternative}
is the numerator: 
\[
i^{N-1}+\left(-i\right)i^{N-2}+\left(-i\right)^{2}i^{N-3}+\dots+\left(-i\right)^{N-2}i+\left(-i\right)^{N-1}=\frac{i^{N}-\left(-i\right)^{N}}{i-\left(-i\right)}=\begin{cases}
\begin{array}{c}
+1\quad N\equiv1\mbox{ }\left(\mbox{mod }4\right)\\
0\quad N\equiv2\mbox{ }\left(\mbox{mod }4\right)\\
-1\quad N\equiv3\mbox{ }\left(\mbox{mod }4\right)
\end{array} & .\end{cases}
\]
When we divide by the denominators $2,4,2$ we find the main diagonal
of $\left(H_{1}^{c}\right)^{-1}$ as the coefficients of $S^{N}=\mathbb{I}$
in $\left(S+S^{-1}\right)^{-1}$: $\frac{1}{2},0,-\frac{1}{2}$ for
$N$: $4k+1$, $4k+2,4k+3$ respectively. 

Now we find the coefficient of $S=S^{N+1}$ in Eq. \eqref{eq:GilAlternative}.
The numerator is: $1+i^{N-1}\left(-i\right)+i^{N-2}\left(-i\right)^{2}+\cdots+i\left(-i\right)^{N-1}=1+\left(i\right)\left(-i\right)\left[i^{N-2}+i^{N-1}\left(-i\right)+\dots+\left(-i\right)^{N-2}\right]$.
Simplifying the numerator we find $1+\frac{i^{N-1}-\left(-i\right)^{N-1}}{i-\left(-i\right)}=1,2,-1$
for $N$: $4k+1$, $4k+2,4k+3$ respectively. 

Dividing by $2,4,2$ in the denominator, we find $\frac{1}{2}$ on
the diagonals $\pm1$ of $\left(S+S^{-1}\right)^{-1}$.

Finally, notice that diagonals $2,3,4,5,\dots$ of $\left(S+S^{-1}\right)^{-1}$
will have \textit{opposite sign} to diagonals $0,1,2,3,\dots$. The
multiplication in the numerator of Eq. \eqref{eq:GilAlternative} gives
a cyclic convolution 
\[
\left(1,i,i^{2},\dots,i^{N-1}\right)\star\left(1,-i,\left(-i\right)^{2},\dots,\left(-i\right)^{N-1}\right)
\]
for the coefficients of $S,S^{2},\dots$ . Because $i^{2}=-1$, the
coefficient of $S^{k+2}$ in the numerator of Eq. \eqref{eq:GilAlternative}
is the negative of the coefficient of $S^{k}$. The denominators are
still $2,4,2$ for $N\equiv1,2,3$. So the pattern in $\left(S+S^{-1}\right)^{-1}=-\mathbf{G}_{1}^{c}$
(starting with the main diagonal) is:
\[
\begin{array}{ccc}
N\equiv4k+1 & \quad & \mbox{diagonals }\frac{1}{2},\frac{1}{2},-\frac{1}{2},-\frac{1}{2}\mbox{ repeated}\\
N\equiv4k+2 & \quad & \mbox{diagonals }0,\frac{1}{2},0,-\frac{1}{2}\mbox{ repeated }\\
N\equiv4k+3 & \quad & \mbox{diagonals }-\frac{1}{2},\frac{1}{2},\frac{1}{2},-\frac{1}{2}\mbox{ repeated.}
\end{array}
\]

This completes the alternative proof.
\end{proof}
Analogous to Prop. \ref{prop:BondAltern}, the proposition below gives
the form of the cyclic Hamiltonian with the special case of alternating
bond strengths and its corresponding Green's function. 
\begin{defn}
The cyclic bond alternating Hamiltonian is defined by ($N$ even)
\[
H_{alt}^{c}=\left[\begin{array}{ccccccc}
0 & \beta &  &  &  &  & \alpha\\
\beta & 0 & \alpha\\
 & \alpha & 0 & \beta\\
 &  & \beta & 0 & \ddots\\
 &  &  & \ddots & \ddots & \alpha\\
&  &  &  & \alpha & 0 & \beta\\
\alpha  &  &  &  &  & \beta & 0
\end{array}\right]
\]
\end{defn}
Comment: As before, this model is not circulant nor has it the Toeplitz
structure.
\begin{prop}
The entries of the Green's function $-\left(H_{alt}^{c}\right)^{-1}$,
are given by 
\begin{equation}
G\left(r,s\right)=-\frac{1}{4}\left\{ \begin{array}{cccc}
\frac{\left(-\alpha/\beta\right)^{\frac{r-s-1}{2}}}{\beta\left[1-\left(-\frac{\alpha}{\beta}\right)^{N/2}\right]}\left[1+\left(-1\right)^{r}\right]\left[1-\left(-1\right)^{s}\right] & + & \frac{\left(-\beta/\alpha\right)^{\frac{r-s-1}{2}}}{\alpha\left[1-\left(-\frac{\beta}{\alpha}\right)^{N/2}\right]}\left[1-\left(-1\right)^{r}\right]\left[1+\left(-1\right)^{s}\right] & \qquad r>s\\
\\
\frac{\left(-\alpha/\beta\right)^{\frac{N+r-s-1}{2}}}{\beta\left[1-\left(-\frac{\alpha}{\beta}\right)^{N/2}\right]}\left[1+\left(-1\right)^{r}\right]\left[1-\left(-1\right)^{s}\right] & + & \frac{\left(-\beta/\alpha\right)^{\frac{N+r-s-1}{2}}}{\alpha\left[1-\left(-\frac{\beta}{\alpha}\right)^{N/2}\right]}\left[1-\left(-1\right)^{r}\right]\left[1+\left(-1\right)^{s}\right] & \qquad r\le s
\end{array}\right.\label{eq:G_bondAlternation}
\end{equation}
\end{prop}
\begin{proof}
We obtain the inverse by solving for $\mathbf{y}$ in $H_{alt}^{c}\mathbf{x}=\mathbf{y}$;
that is, we think of $\mathbf{y}$ as given and we solve for $\mathbf{x}$.
This will give us $\mathbf{x}=\left(H_{alt}^{c}\right)^{-1}\mathbf{y}$.
First we solve the even rows in terms of the last row $x_{N}$, which
itself can be solved from $x_{N}=\sum_{i=1}^{N}\left[\left(H_{alt}^{c}\right)^{-1}\right]_{N,i}y_{i}$
to give
\begin{eqnarray*}
x_{2k} & = & \frac{1}{\beta}\{\sum_{m=0}^{k-1}\left(-\frac{\alpha}{\beta}\right)^{m}y_{2k-2m-1}\}+\left(-\frac{\alpha}{\beta}\right)^{k}x_{N}\\
x_{N} & = & \beta^{-1}\left[1-\left(-\alpha/\beta\right)^{N/2}\right]^{-1}\mbox{ }\sum_{m=0}^{\frac{N}{2}-1}\left(-\frac{\alpha}{\beta}\right)^{m}y_{N-2m-1}
\end{eqnarray*}

Similarly the odd rows are obtained in terms of $x_{1}$, which itself
can be solved $x_{1}=\sum_{i=1}^{N}\left[\left(H_{alt}^{c}\right)^{-1}\right]_{1,i}y_{i}$
to give
\begin{eqnarray*}
x_{2k+1} & = & \frac{1}{\alpha}\{\sum_{m=0}^{k-1}\left(-\frac{\beta}{\alpha}\right)^{m}y_{2k-2m}\}+\left(-\frac{\beta}{\alpha}\right)^{k}x_{1}\\
x_{1} & = & \alpha^{-1}\left[1-\left(-\beta/\alpha\right)^{N/2}\right]^{-1}\{y_{N}-\frac{\beta}{\alpha}\mbox{ }\sum_{m=0}^{\frac{N}{2}-1}\left(-\frac{\beta}{\alpha}\right)^{m}y_{N-2m-2}\}
\end{eqnarray*}

Combining these equations to solve for the even and odd rows separately
and multiplying by an overall minus sign we arrive at $\mathbf{G}=-\left(H_{alt}^{c}\right)^{-1}$
given by Eq. \eqref{eq:G_bondAlternation}.
\end{proof}
Comment: In the special case that $N=2k\left(2k-1\right)$, $G_{1}^{c}$
in Proposition \ref{prop:Cycle-Green's-function} can be obtained
from Eq. \eqref{eq:G_bondAlternation} by substituting $\alpha=\beta=1$.
Note that in this limit, it is necessary that $N\ne4k$ for the denominator
not to vanish in agreement with Lemma \ref{lem:Cycle-determinat}.

We now pose a more general (and difficult) question. When does the
inverse exist in spatial dimension $d$ and if it does, how can it
be computed? In the next section we use mathematical techniques borrowed
from quantum information theory and number theory to address some
of these problems.
\section{Higher Dimensional Green's Function}
The Green's function we derived is the negative of the inverse of
the H{\"u}ckel (tight binding) Hamiltonian, whose $N\times N$ matrix
representation in Dirac notation \cite{Dirac1967} is 
\begin{equation}
H_{1}=\sum_{k=1}^{N}\left\{ \mbox{ }|k\rangle\langle k+1|\mbox{ }+\mbox{ }|k+1\rangle\langle k|\mbox{ }\right\} \quad,\label{eq:H1D-1}
\end{equation}
where in units of $\beta$ the coupling can be taken to be one. 

To explore the $d-$dimensional analog $H_{d}$, we use tensor products
of matrices. Recall that the tensor product of an $m\times n$ matrix
$A$ and an $p\times q$ matrix $B$ is the $mp\times nq$ matrix
defined by 
\[
A\otimes B=\left[\begin{array}{cccc}
a_{11}B & a_{12}B & \cdots & a_{1n}B\\
a_{21}B & a_{22}B & \cdots & a_{2n}B\\
\vdots & \vdots &  & \vdots\\
a_{m1}B & a_{m2}B & \cdots & a_{mn}B
\end{array}\right].
\]

The Hamiltonian, $H_{d}$, on a square lattice in $d-$spatial dimensions
(square lattice in $d=2$, cubic in $d=3$, etc.), with the linear
size $N$ can succinctly be expressed as 
\begin{equation}
H_{d}=\sum_{i=1}^{d}\mathbb{I}_{N^{i-1}}\otimes H_{1}\otimes\mathbb{I}_{N^{d-i}}\label{eq:Hd}
\end{equation}
where $H_{1}$ is given by Eq. \eqref{eq:H1D-1}, and the size of every
identity matrix is indicated by its subscript. In dimensions $2$
and $3$, the Hamiltonians come from $H_{1}$ and $\mathbb{I}=\mathbb{I}_{N}$:
\begin{eqnarray}
H_{2} & = & \left(H_{1}\otimes\mathbb{I}\right)+\left(\mathbb{I}\otimes H_{1}\right)\label{eq:H2D}\\
H_{3} & = & \left(H_{1}\otimes\mathbb{I}\otimes\mathbb{I}\right)+\left(\mathbb{I}\otimes H_{1}\otimes\mathbb{I}\right)+\left(\mathbb{I}\otimes\mathbb{I}\otimes H_{1}\right)\quad.\label{eq:H3D}
\end{eqnarray}

Comment: The techniques apply more generally where the lattice can
be constructed from $d$ independent linear subsets.

Comment: When $H_{1}$ is a Toeplitz or a circulant matrix, the corresponding
$H_{d}$ is generally not a Toeplitz or a circulant matrix \cite{meurant1992review},
but they will be block Toeplitz or block circulant respectively.

The eigenvalue decomposition of $H_{1}=Q\Lambda Q^{T}$, where $\Lambda$
is the $N\times N$ diagonal matrix of eigenvalues whose $k^{\mbox{th}}$
entry is $2\cos k\omega$ and $Q$ is the matrix of eigenvectors with
$r^{\mbox{th}}$ column given by Eq. \eqref{eq:Evec}. Since $\cos k\omega\ne0$
for all $1\le k\le N$, $\Lambda$ is a diagonal matrix with no zero
entries on the diagonal and $H_{1}$ is invertible, i.e., has a Green's
function, as expected from our calculations.

The associated Green's function matrix in $d$ dimensions is defined
by $\mathbf{G}_{d}=-H_{d}^{-1}$. Obtaining an analytical expression
for the inverse in higher dimensions, at first, might seem difficult
because it involves sums of matrices. In $d=2$ the size of the lattice
is $N\times N$ and in $d=3$ the size is $N\times N\times N$. 

After the eigenvalue decomposition, the Hamiltonians in higher dimensions
(e.g., Eqs. \eqref{eq:H2D},\eqref{eq:H3D}) reads
\begin{eqnarray}
H_{d} & = & Q^{\otimes d}\left\{ \sum_{i=1}^{d}\mathbb{I}_{N^{i-1}}\otimes\Lambda\otimes\mathbb{I}_{N^{d-i}}\right\} \left(Q^{T}\right)^{\otimes d}\label{eq:Hd_EigDecom}\\
 & \equiv & Q^{\otimes d}\left\{ \quad\Lambda_{d}\quad\right\} \left(Q^{T}\right)^{\otimes d}\label{eq:Hd_EigDecom_Defined}
\end{eqnarray}
where the matrix of eigenvectors denoted by $Q^{\otimes d}\equiv Q\otimes\cdots\otimes Q$
is a $d$-fold tensor product and the diagonal matrix of eigenvalues
is $\Lambda_{d}=\sum_{i=1}^{d}\mathbb{I}_{N^{i-1}}\otimes\Lambda\otimes\mathbb{I}_{N^{d-i}}$.
For example,
\begin{eqnarray}
H_{2} & = & \left(Q\Lambda Q^{T}\otimes\mathbb{I}\right)+\left(\mathbb{I}\otimes Q\Lambda Q^{T}\right)\label{eq:H2D-1}\\
 & = & \left(Q\otimes Q\right)\mbox{ }\left[\left(\Lambda\otimes\mathbb{I}\right)+\left(\mathbb{I}\otimes\Lambda\right)\right]\mbox{ }\left(Q^{T}\otimes Q^{T}\right)\nonumber \\
H_{3} & = & \left(Q\Lambda Q^{T}\otimes\mathbb{I}\otimes\mathbb{I}\right)+\left(\mathbb{I}\otimes Q\Lambda Q^{T}\otimes\mathbb{I}\right)+\left(\mathbb{I}\otimes\mathbb{I}\otimes Q\Lambda Q^{T}\right)\label{eq:H3D-1}\\
 & = & \left(Q\otimes Q\otimes Q\right)\mbox{ }\left[\left(\Lambda\otimes\mathbb{I}\otimes\mathbb{I}\right)+\left(\mathbb{I}\otimes\Lambda\otimes\mathbb{I}\right)+\left(\mathbb{I}\otimes\mathbb{I}\otimes\Lambda\right)\right]\mbox{ }\left(Q^{T}\otimes Q^{T}\otimes Q^{T}\right).\nonumber 
\end{eqnarray}

This change of basis allows us to diagonalize the Hamiltonians in
any dimension, for example
\begin{eqnarray}
\Lambda_{2} & = & \left(\Lambda\otimes\mathbb{I}\right)+\left(\mathbb{I}\otimes\Lambda\right)\label{eq:H2D_diag}\\
\Lambda_{3} & = & \left(\Lambda\otimes\mathbb{I}\otimes\mathbb{I}\right)+\left(\mathbb{I}\otimes\Lambda\otimes\mathbb{I}\right)+\left(\mathbb{I}\otimes\mathbb{I}\otimes\Lambda\right).\label{eq:H3D_diag}
\end{eqnarray}

Below we investigate the conditions under which the Green's function
exists. For now suppose that it does. Its algebraic representation
in $d$ dimensions (compare with Eqs. \eqref{eq:Hd_EigDecom} and \eqref{eq:Hd_EigDecom_Defined})
is 
\begin{equation}
\mathbf{G}_{d}=-Q^{\otimes d}\mbox{ }\left\{ \sum_{i=1}^{d}\mathbb{I}_{N^{i-1}}\otimes\Lambda\otimes\mathbb{I}_{N^{d-i}}\right\} ^{-1}\mbox{ }\left(Q^{T}\right)^{\otimes d}.\label{eq:Gd}
\end{equation}
It is clear that if $\mathbf{G}_{d}$ were to exist no eigenvalue
can be zero. Namely, diagonal entries being all the possible sums
should satisfy $2\sum_{i=1}^{d}\cos\left(\frac{k_{i}\pi}{N+1}\right)\ne0$
for any choice of $1\le k_{i}\le N$.  Then, the corresponding eigenvalues
of $\mathbf{G}_{d}$ are $-1/\left\{ 2\sum_{i=1}^{d}\cos\left(\frac{k_{i}\pi}{N+1}\right)\right\} $.

As an illustration let us take $d=2$. Then the energies are the diagonal
entries of $\Lambda_{2}$ given by the sum 
\[
\Lambda_{2}=2\left[\begin{array}{cccc}
\left(\cos\omega\right)\mathbb{I}\\
 & \left(\cos2\omega\right)\mathbb{I}\\
 &  & \ddots\\
 &  &  & \left(\cos N\omega\right)\mathbb{I}
\end{array}\right]+\left[\begin{array}{cccc}
\Lambda\\
 & \Lambda\\
 &  & \ddots\\
 &  &  & \Lambda
\end{array}\right],
\]
which is a matrix of size $N^{2}\times N^{2}$; $\Lambda$ and $\left(\cos k\omega\right)\mathbb{I}$
are ($N\times N$). Since $\cos k\omega=-\cos\left[\left(N+1-k\right)\omega\right]$,
each block of the sum is $2\left(\cos k\omega\right)\mathbb{I}+\Lambda$
for some $1\le k\le N$ whose $\left(N+1-k\right)^{\mbox{th}}$ entry
is zero. Therefore, the diagonal $N^{2}\times N^{2}$ matrix $\Lambda_{2}$
has exactly $N$ zeros on its diagonal, one in each of the $N$ blocks,
and hence noninvertible.

\section{Green's function and number theory}
The existence of the Green's function, $\mathbf{G}_{d}$ , in higher
dimensions requires that $H_{d}$ has non-zero eigenvalues, i.e.,
$\sum_{i=1}^{d}\cos\left(\frac{k_{i}\pi}{N+1}\right)\ne0$ for any
choice of $1\le k_{i}\le N$. 
\begin{lem}
\label{The-Green's-function_Existence}$H_{d}^{-1}$ does not exist
in even spatial dimensions.\end{lem}
\begin{proof}
Since $\cos k\omega=-\cos\left[\left(N+1-k\right)\omega\right]$ for
any $1\le k\le N$, we can always pair up the cosines such that each
pair sums to zero implying that there is a zero eigenvalue.
\end{proof}
Therefore, below we take $d$ and $N+1$ to be odd (as $N$ odd is
already non-invertible in one dimension). 

We need to prove the general conditions under which $H_{d}$ is invertible,
which is a problem in number theory. Recently there has been quite
a bit of interest in a closely related question, which is under what
conditions do sums of roots of unity vanish? Besides sheer theoretical
interest, this problem is related to many mathematical structures.
For example, Poonen and Rubinstein relate this problem to the number
of interior intersection points made by the diagonals of a regular
$n-$gon \cite{Poonen1998}. 

Let us denote $n\equiv N+1$. Suppose one asks for what natural numbers
$d$ do there exist $n^{th}$ roots of unity $\alpha_{1},\dots,\alpha_{d}\in\mathbb{C}$
such that $\alpha_{1}+\cdots+\alpha_{d}=0$? Such an equation is said
to be a vanishing sum of $n^{th}$ roots of unity of \textit{weight
$d$}. Let $n$ have the prime factorization $p_{1}^{a_{1}}\cdots p_{r}^{a_{r}}$
($a_{i}>0$), then we can define $W\left(n\right)$ to be the set
of weights $d$ for which there exists a vanishing sum $\alpha_{1}+\cdots+\mbox{\ensuremath{\alpha}}_{d}=0$;
if the sum does not vanish then $W\left(n\right)$ is simply the empty
set. 

Before delving into the proof we introduce some notation and terminology
presented in \cite{Lam2000}. Let $\langle G\rangle$ be a cyclic
group of order $n$ and let $\zeta$ be a (fixed) primitive $n^{th}$
root of unity. There exists a natural ring homomorphism $\varphi$
from the integral group $\mathbb{Z}G$ to the ring of cyclotomic integers
$\mathbb{Z}\left[\zeta\right]$, given by the equation $\varphi\left(z\right)=\zeta$,
i.e., the map $\varphi:\mathbb{Z}G\rightarrow\mathbb{Z}\left[\zeta\right]$.
An element of $\mathbb{Z}G$, say $x=\sum_{g\in G}x_{g}g$, lies in
the kernel $\mbox{ker}\left(\varphi\right)$ if and only if $\sum_{g\in G}x_{g}\varphi\left(g\right)=0$
in $\mathbb{Z}\left[\zeta\right]$. Therefore, the elements of the
ideal $\mbox{ker}\left(\varphi\right)$ correspond precisely to all
$\mathbb{Z}$-linear relations among the $n^{th}$ roots of unity.
For vanishing sums of $n^{th}$ roots of unity, we have to look at
elements $x=\sum_{g}x_{g}g\in\mbox{ker}\left(\varphi\right)$ with
$x_{g}\ge0$; the number of non-zero coefficients $x_{g}$ is denoted
by $\epsilon_{0}\left(x\right)$. In other words one looks at $\mathbb{N}G\cap\mbox{ker}\left(\varphi\right)$,
where $\mathbb{N}G$ denotes the group semi-ring of $G$ over $\mathbb{N}$. 

A vanishing sum $\alpha_{1}+\cdots+\alpha_{d}=0$ is called \textit{minimal}
if no proper sub-sum is zero. Clearly, one can always multiply a vanishing
sum by a root of unity to get another vanishing sum; we say the latter
is\textit{ similar} to the former; i.e., one can be obtained from
the other by a rotation. For any natural number $n$, $\zeta_{n}$
denotes a primitive $n^{th}$ root of unity in $\mathbb{C}$. 

In terms of roots of unity, a vanishing sum from the basic relations
of the form 
\begin{equation}
1+\zeta_{p_{i}}+\zeta_{p_{i}}^{2}+\cdots+\zeta_{p_{i}}^{p_{i}-1}=0\qquad1\le i\le r\label{eq:symmetric}
\end{equation}
is called a \textit{symmetric} minimal elements in $\mathbb{N}G\cap\mbox{ker}\left(\varphi\right)$.
In general, there are vanishing minimal sums which are not similar
to those in Eq. \eqref{eq:symmetric}. The latter are called \textit{asymmetric} sums.

The following theorem due to Lam and Leung \ref{=00005BLam-and-Leung,}
\cite{Lam2000}, will help us prove our theorem pertaining to vanishing
sums of cosines (Theorem \ref{MAINTheorem_SumsCosines}).
\begin{thm*}
\label{=00005BLam-and-Leung,}{[}Lam and Leung, Theorem 4.8{]} Let
$G$ be a cyclic group of order $n=p_{1}^{a_{1}}p_{2}^{a_{2}}\cdots p_{r}^{a_{r}}$
, where $p_{1}<p_{2}<\cdots<p_{r}$ are primes and let $\varphi:\mathbb{Z}G\rightarrow\mathbb{Z}\left[\zeta\right]$
be as above, where $\zeta=\zeta_{n}$. For any minimal element $x\in\mathbb{N}G\cap\mbox{ker}\left(\varphi\right)$,
we have either (A) $x$ is symmetric, or (B) $r\ge3$ and $\epsilon_{0}\left(x\right)\ge p_{1}\left(p_{2}-1\right)+p_{3}-p_{2}>p_{3}$. 
\end{thm*}
We shall utilize this theorem to prove the following (recall that
$n=N+1$):
\begin{thm}
\label{MAINTheorem_SumsCosines}Let $n$ be a positive odd integer
and $k_{1}$, $k_{2}$, $\cdots$, $k_{d}$ be a set of integers such
that $1\le k_{i}\le n-1$. Then 
\begin{equation}
\sum_{i=1}^{d}\cos\left(\frac{k_{i}\pi}{n}\right)\ne0\label{eq:sumCosines}
\end{equation}
for any choice of $k_{i}$'s if and only if  $d$ is odd and is smaller
than the smallest divisor of $n$.\end{thm}
\begin{proof}
By Lemma \ref{The-Green's-function_Existence}, we only need to consider
$d$ odd. Below we first work with roots of unity by writing the cosines
in terms of the roots 
\begin{equation}
\sum_{i=1}^{d}\cos\left(\frac{2k_{i}\pi}{2n}\right)=2\left\{ \sum_{i=1}^{d}\zeta_{2n}^{k_{i}}+\zeta_{2n}^{-k_{i}}\right\} .\label{eq:sumsOfCosines_MainThm}
\end{equation}
So we have now a sum over $2d$ roots of unity. We first prove that
this sum is never zero if $d<p_{2}$. Since $2n=2p_{2}^{a_{2}}p_{3}^{a_{3}}\cdots p_{r}^{a_{r}}$
with all the $p_{i}$'s being odd, we are guaranteed (from Theorem
\ref{=00005BLam-and-Leung,}) that $p_{1}\left(p_{2}-1\right)+p_{3}-p_{2}=p_{2}+p_{3}-2<2p_{2}$
therefore $2d\equiv\epsilon_{0}\left(x\right)<2p_{2}$ and if there
were vanishing sums they would be of type (A), which are symmetric,
i.e., sums of minimal relations. When $d<p_{2}$, in Eq. \eqref{eq:sumsOfCosines_MainThm}
there would be fewer than $2p_{2}$ points on the unit circle all
of which appear as complex conjugate pairs. For the sum to be of type
(A) and vanish, there should be a symmetric sum with a prime $p$
that vanishes. The corresponding roots are a subset of the original
points that are a vanishing sum of roots of $\zeta_{p}$ with the
prime $p\ge p_{2}>d$, therefore it would involve a vanishing sum
on more than half of the points of the original $2d$ terms in Eq.
\ref{eq:sumsOfCosines_MainThm}. Hence there must be at least one
complex conjugate pair in the vanishing sum under consideration. But
if there is one complex conjugate pair then all the roots should be
complex conjugates as we can rotate any of the $p^{\mbox{th}}$ roots
into one another. Since we have a vanishing sum of complex conjugate
pairs but we allow only an odd number of terms there must be a real
root. But we exclude the real roots ($\pm1$). Therefore we reach
a contradiction and the sum can never vanish. 
\begin{figure}
\begin{centering}
\includegraphics[scale=0.4]{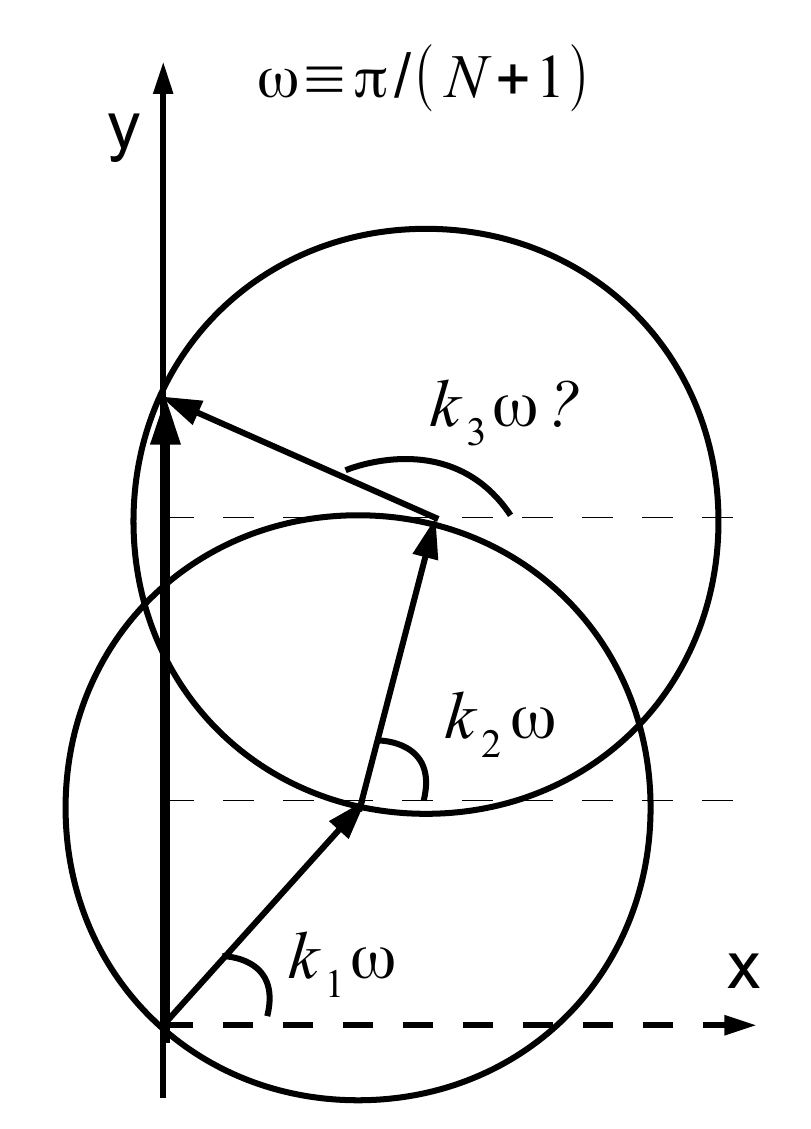}$\qquad\qquad\qquad$\includegraphics[scale=0.4]{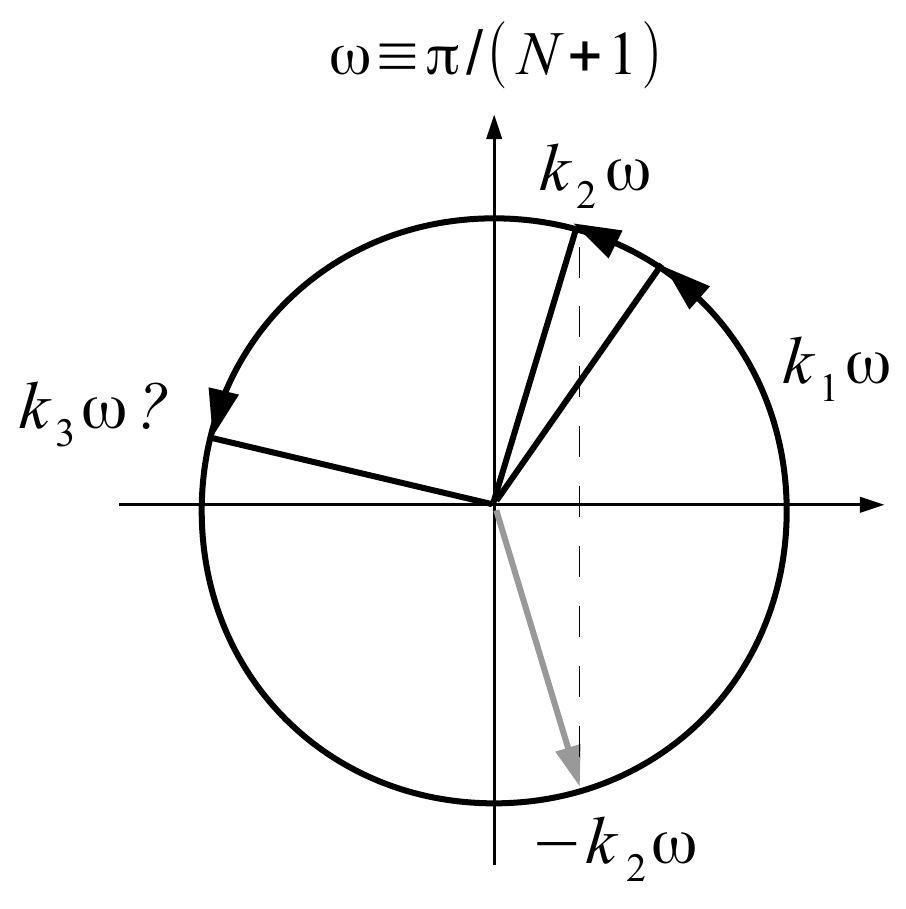}\protect\caption{\label{fig:SumofCosines}Left: $\cos\left(\frac{k_{1}\pi}{N+1}\right)+\cos\left(\frac{k_{2}\pi}{N+1}\right)+\cos\left(\frac{k_{3}\pi}{N+1}\right)=0$
is equivalent to the phasors adding to a vertical vector. The circles
shown are unit circles. Right: Vanishing sums of roots of unity imply
vanishing sums of cosines on the upper half plane since one can reflect
any phasor without changing the cosine.}
\par\end{centering}
\end{figure}

Now we prove that the sum can be zero if $d\ge p_{2}$. It is sufficient
to show that it vanishes for $d=p_{2}$ as for any odd $d>p_{2}$
we can always pair up the $2\left(p_{2}-d\right)$ cosines to cancel
as we did in the proof of Lemma \ref{The-Green's-function_Existence}.
Suppose $d=p_{2}$. Then a symmetric sum over the roots of unity that
vanishes implies that the sum over cosines vanishes as the cosines
are the real part and geometrically one can reflect the roots to the
upper half plane (see Fig. \ref{fig:SumofCosines}) . However, we
need to exclude the possibility of $\pm1$ as roots and show that
the sum still vanishes. The number of symmetric sums will be $\frac{2n}{p}$
but only $2$ of them have $\pm1$ as roots. In the sum involving
the symmetric sums we can exclude the ones that have $\pm1$ and still
be left with vanishing symmetric sums. \end{proof}
\begin{cor}
The inverse of the H{\"u}ckel matrix and hence its Green's function
in $d$ dimensions exists if and only if $d$ is odd and is smaller
than the smallest prime divisor of  $N+1$. 
\end{cor}
For $d=3$, this lemma and Eq. \eqref{eq:sumCosines} have the geometrical
interpretation shown in Fig. \ref{fig:SumofCosines}. Moreover, twice
the left hand side of Eq. \eqref{eq:sumCosines} is the expression for
the energies of the H{\"u}ckel matrix in $d$ dimensions. 

\section{The physical consequences of the inverse of the Hückel matrix and
the zeroes of its Green\textquoteright s function}
The H{\"u}ckel formalism, in its physical and chemical context, is
not, of course, restricted to a linear chain. Various two- and three-dimensional
connectivities have been probed in the $80$ years of its existence,
to the immense benefit of practice and understanding in chemistry.
But until recent time, there has been scant interest in the Green\textquoteright s
function of the H{\"u}ckel matrix, and its inverse. Heilbronner used
the inverse of the H{\"u}ckel Matrix to form an undervalued bridge
between the resonance structure of valence bond theory and molecular
orbitals - thus bringing together two seemingly distinct, but in fact
related, approaches to the electronic structure of molecules \cite{Heilbronner1976}.
The graph theoretical context has led people to investigate the inverse
of the vertex adjacency matrix \cite{Farrugia2013}. In the work of
Estrada, the relationship between the Green\textquoteright s function
formalism and the inverse of the vertex adjacency matrix of a graph
is consistently utilized \cite{Estrade2007,Estrade2008}.

In a field that has attracted much attention both experimentally and
theoretically in the last decade, the transmission of current across
molecules, a striking phenomenon, quite nonclassical, is observed.
This is quantum interference, zero or low conductance when electrodes
are attached to specific sites across a molecule \cite{Datta2005,Solomon2013}.
Quantum interference occurs when the Green\textquoteright s function,
whose absolute value squared is related to the current transmitted,
vanishes. These are exactly the zeroes of Eq. \eqref{eq:Math}. The
inverse of the H{\"u}ckel matrix has been directly related to this
phenomenon in the work of Markussen and Stadler \cite{Markussen2011}.
The chemical consequences of just these zeroes have been outlined
in recent work by us \cite{tsuji2014quantum}. 

The results we have obtained in this paper for the specific entries of the Green's function have been proven of great utility in describing the transmission of current through molecules. The Green's function expressions obtained in this paper also play an important role in designing a new molecular switches based on electrocyclic reactions, cycloadditions, and sigmatropic reactions in linear polyenes \cite{Solomon2013,tsuji2014quantum}. In particular in a paper on linear polyenes, we have used the results to derive specifically the transmission across a chain, and its exponential falloff \cite{tsuji2015exponential}. In the same paper, where it was important to have the Green's function elements for a cyclic polyene (annulene) with and without bond alternation, expressions from the current paper and some based on similar procedures, were used. Two kinds of zero values of the Green's function introduced in this paper, namely easy zero and hard zero, provide the bases for an important classification of quantum interference, and the connection between these zeroes and the non-disjoint or disjoint nature of diradical molecules has been clarified \cite{tsuji2016close}. 

\section{Acknowledgements}
We thank Peter Shor for his help in proving Theorem \ref{MAINTheorem_SumsCosines}.
We also thank Christopher King, Jeffrey Goldstone, and Shev MacNamara. RM and RH thank National Science Foundation through
Grants DMS 1312831 and CHE 1305872 respectively, G.S. is grateful
for the ongoing support of MathWorks, and Y. T. thanks the Japan Society
for the Promotion of Science for a JSPS Postdoctoral Fellowship for
Research Abroad. RM thanks IBM Research for the support and the freedom provided by the Herman Goldstine fellowship.
\bibliographystyle{plain}
\bibliography{mybib}
\newpage{}

\section{Appendix: First principle derivation from trigonometrics}
\begin{prop*}
The following sum has a closed form solution 

\begin{equation}
-\frac{1}{N+1}\sum_{k=1}^{N}\frac{\sin\left(rk\omega\right)\sin\left(sk\omega\right)}{\cos\left(k\omega\right)}=e^{i\pi\left(\frac{r+s-1}{2}\right)}\label{eq:Math-1}
\end{equation}
when $r$ is even and $s<r$ is odd. Otherwise, when $s\le r$ it
is zero. Moreover $s>r$ are symmetric, i.e., $G\left(r,s\right)=G\left(s,r\right)$.
\end{prop*}

\subsection{Easy Zeros: Same parity of $s$ and $r$}

Using $\sin\left(rk\omega\right)\sin\left(sk\omega\right)=\frac{1}{2}\left\{ \cos\left[\left(r-s\right)k\omega\right]-\cos\left[\left(r+s\right)k\omega\right]\right\} $,
we can rewrite Eq. \eqref{eq:Problem} as 
\begin{equation}
G\left(r,s\right)=-\frac{1}{2\left(N+1\right)}\sum_{k=1}^{N}\frac{\left\{ \cos\left[\left(r-s\right)k\omega\right]-\cos\left[\left(r+s\right)k\omega\right]\right\} }{\cos\left(k\omega\right)}\label{eq:Problem_Cosined}
\end{equation}

First let $r-s=2q$, this implies that $r$ and $s$ have the same
parity (i.e., oddness or evenness). Therefore $r+s$ is also even,
let it be $r+s=2q'$ for some $q'\in\mathbb{N}$. Eq. \eqref{eq:Problem_Cosined}
becomes

\begin{equation}
G\left(r,s\right)=-\frac{1}{2\left(N+1\right)}\sum_{k=1}^{N}\left\{ \frac{\cos\left[2qk\omega\right]}{\cos\left(k\omega\right)}-\frac{\cos\left[2q'k\omega\right]}{\cos\left(k\omega\right)}\right\} .\label{eq:problem1}
\end{equation}
We now show that each sum is zero. Let us first show $\sum_{k=1}^{N}\frac{\cos\left[2qk\omega\right]}{\cos\left(k\omega\right)}=0$.
Recall $\omega=\frac{\pi}{N+1}$ and expand the sum by adding the
first to the last then the second to $N-1^{\mbox{st}}$ etc. to get

\begin{eqnarray}
\sum_{k=1}^{N}\frac{\cos\left[2qk\omega\right]}{\cos\left(k\omega\right)} & = & \left(\frac{\cos\left[2q\omega\right]}{\cos\left(\omega\right)}+\frac{\cos\left[2qN\omega\right]}{\cos\left(N\omega\right)}\right)+\left(\frac{\cos\left[4q\omega\right]}{\cos\left(2\omega\right)}+\frac{\cos\left[2q\left(N-1\right)\omega\right]}{\cos\left(\left(N-1\right)\omega\right)}\right)\label{eq:Proof1}\\
 & + & \cdots+\left(\frac{\cos\left[2q\frac{N}{2}\omega\right]}{\cos\left(\omega N/2\right)}+\frac{\cos\left[2q\left(\frac{N}{2}+1\right)\omega\right]}{\cos\left(\left(\frac{N}{2}+1\right)\omega\right)}\right).\nonumber 
\end{eqnarray}
We now show that each of the parenthesis is identically zero. To do
so we notice that each of the parenthesis is of the form
\[
\frac{\cos\left[2qk\omega\right]}{\cos\left(k\omega\right)}+\frac{\cos\left[2q\left(N-k+1\right)\omega\right]}{\cos\left(\left(N-k+1\right)\omega\right)};\qquad k=1,2,\dots,\frac{N}{2}.
\]
But $\cos\left(\left(N-k+1\right)\frac{\pi}{N+1}\right)=\cos\left(-k\frac{\pi}{N+1}+\pi\right)=-\cos\left(k\frac{\pi}{N+1}\right)=-\cos\left(k\omega\right)$
by the double angle formula and evenness of the cosine. Moreover
\begin{align*}
\cos\left[2q\left(N-k+1\right)\omega\right] & =\cos\left[2q\pi-2q\frac{k\pi}{N+1}\right]=\cos\left[-2q\frac{k\pi}{N+1}\right]\\
 & =\cos\left[2q\frac{k\pi}{N+1}\right]=\cos\left(2qk\omega\right)
\end{align*}
Concluding that the numerators are equal but denominators differ in
sign resulting in 
\begin{align*}
\frac{\cos\left[2qk\omega\right]}{\cos\left(k\omega\right)}+\frac{\cos\left[2q\left(N-k+1\right)\omega\right]}{\cos\left(\left(N-k+1\right)\omega\right)} & =\frac{\cos\left[2qk\omega\right]}{\cos\left(k\omega\right)}-\frac{\cos\left[2qk\omega\right]}{\cos\left(k\omega\right)}=0\quad.
\end{align*}
The exact same argument with substitution $q'$ for $q$ in Eq. \eqref{eq:Proof1}
proves that the second sum in Eq. \eqref{eq:problem1} is zero. Together
proving $G\left(r,s\right)=0$ if $r$ and $s$ have the same parity.
There are other zeros that are harder to prove.

\subsection{Harder Zeros}

Let us make the sum in Eq. \eqref{eq:Problem} centered by letting $m=k-\frac{N+1}{2}$,
whereby

\[
G\left(r,s\right)=-\frac{1}{N+1}\sum_{m=-\frac{N-1}{2}}^{\frac{N-1}{2}}\frac{\sin\left[r\omega\left(m+\frac{N+1}{2}\right)\right]\sin\left[s\omega\left(m+\frac{N+1}{2}\right)\right]}{\cos\left[\omega\left(m+\frac{N+1}{2}\right)\right]}
\]
where $\omega=\frac{\pi}{N+1}$ as before. Since, $\cos\left[\omega\left(m+\frac{N+1}{2}\right)\right]=-\sin\left(\omega m\right)$
and $\sin x=\frac{1}{2i}\left(e^{ix}-e^{-ix}\right)$ 

\begin{eqnarray}
G\left(r,s\right) & = & -\frac{ie^{i\left(r+s\right)\pi/2}}{2\left(N+1\right)}\sum_{m=-\frac{N-1}{2}}^{\frac{N-1}{2}}e^{i\omega m\left(r+s-1\right)}\frac{\left(1-e^{-2ir\omega\left(m+\frac{N+1}{2}\right)}\right)\left(1-e^{-2is\omega\left(m+\frac{N+1}{2}\right)}\right)}{1-e^{-2i\omega m}}\nonumber \\
 & = & -\frac{ie^{i\left(r+s\right)\pi/2}}{2\left(N+1\right)}\sum_{m=-\frac{N-1}{2}}^{\frac{N-1}{2}}e^{i\omega m\left(r+s-1\right)}\frac{\left[1-\left(-1\right)^{r}e^{-2ir\omega m}\right]\left[1-\left(-1\right)^{s}e^{-2is\omega m}\right]}{1-e^{-2i\omega m}}\label{eq:Gmassaged}
\end{eqnarray}

This equation is general and will be used later for nonzero sums as
well. \\

Since we proved that if $r$ and $s$ have the same parity the sum
vanishes, we prove the harder zeros (see Eq. \eqref{eq:The-Green's-function})
by letting $r$ be odd and $s$ even and enforcing $s<r$. We can
let $r=2q-1$ and $s=2p$ with integers $p$ and $q$ satisfying $0<p<q\le N/2$.
Using these, Eq. \eqref{eq:Gmassaged} becomes
\[
G\left(r,s\right)=-\frac{e^{i\left(p+q\right)\pi}}{2\left(N+1\right)}\sum_{m=-\frac{N-1}{2}}^{\frac{N-1}{2}}e^{i2\omega m\left(p+q-1\right)}\frac{\left[1+e^{-2i\left(2q-1\right)\omega m}\right]\left[1-e^{-2i\left(2p\right)\omega m}\right]}{1-e^{-2i\omega m}}
\]

We now use the factorization $\frac{1-x^{2\ell}}{1-x}=1+x+x^{2}+\cdots+x^{2\ell-1}$
with $x\equiv\exp\left(-i2m\omega\right)$ to get rid of the denominator

\begin{eqnarray*}
G\left(r,s\right) & = & -\frac{\left(-1\right)^{p+q}}{2\left(N+1\right)}\sum_{m=-\frac{N-1}{2}}^{\frac{N-1}{2}}\left\{ e^{2i\omega m\left(p+q-1\right)}\left[1+e^{-2i\omega m\left(2q-1\right)}\right]\right.\\
 & \times & \left.\left[1+e^{-2i\omega m}+e^{-4i\omega m}+\cdots+e^{-2i\omega m\left(2p-1\right)}\right]\right\} 
\end{eqnarray*}

Multiplying the phase factor into the parenthesis inside the sum and
substituting for $\omega$ we have
\begin{equation}
G\left(r,s\right)=-\frac{\left(-1\right)^{p+q}}{2\left(N+1\right)}\sum_{m=-\frac{N-1}{2}}^{\frac{N-1}{2}}\left(e^{2i\left(p+q-1\right)\frac{\pi m}{N+1}}+e^{-2i\left(q-p\right)\frac{\pi m}{N+1}}\right)\left\{ 1+e^{-2i\frac{\pi m}{N+1}}+e^{-4i\frac{\pi m}{N+1}}+\cdots+e^{-2i\left(2p-1\right)\frac{\pi m}{N+1}}\right\} .\label{eq:HardZero_Final}
\end{equation}
Comment: The pre-factor multiplying the sum can only be $\frac{\pm1}{2\left(N+1\right)}$,
determined by the values of $p$ and $q$ : $G\left(r,s\right)$ vanishes
iff the sum does.\\

We expand the summand in Eq. \eqref{eq:HardZero_Final} to get
\begin{eqnarray}
G\left(r,s\right) & = & -\frac{\left(-1\right)^{p+q}}{2\left(N+1\right)}\sum_{m=-\frac{N-1}{2}}^{\frac{N-1}{2}}\left\{ \left[e^{2i\left(p+q-1\right)\frac{\pi m}{N+1}}+e^{2i\left(p+q-2\right)\frac{\pi m}{N+1}}+\cdots+e^{2i\left(q-p+1\right)\frac{\pi m}{N+1}}+e^{2i\left(q-p\right)\frac{\pi m}{N+1}}\right]\right.\nonumber \\
 & + & \left.\left[e^{-2i\left(q-p\right)\frac{\pi m}{N+1}}+e^{-2i\left(q-p+1\right)\frac{\pi m}{N+1}}+e^{-2i\left(q-p+2\right)\frac{\pi m}{N+1}}+\cdots+e^{-2i\left(p+q-2\right)\frac{\pi m}{N+1}}+e^{-2i\left(p+q-1\right)\frac{\pi m}{N+1}}\right]\right\} \nonumber \\
 & = & -\frac{\left(-1\right)^{p+q}}{\left(N+1\right)}\sum_{m=-\frac{N-1}{2}}^{\frac{N-1}{2}}\left\{ \cos\left[\frac{2\pi m\left(q-p\right)}{N+1}\right]+\cos\left[\frac{2\pi m\left(q-p+1\right)}{N+1}\right]+\cdots+\cos\left[\frac{2\pi m\left(p+q-1\right)}{N+1}\right]\right\} \label{eq:HardZeros_CosineForm}
\end{eqnarray}
where in the last equation, to get the cosines, we paired the first
term inside the first brackets with the last term inside the second
brackets etc. and used the formula $e^{ix}+e^{-ix}=2\cos x$. The
factor of $2$ cancelled the overall pre-factor $1/2$. \\

Comment: It is important to note that, since $q>p$, the exponents
in the first bracket are all positive and in the second bracket the
exponents are all negative. \\

We can write a more succinct expression

\begin{equation}
G\left(r,s\right)=-\frac{\left(-1\right)^{p+q}}{\left(N+1\right)}\mbox{ }\sum_{t=0}^{2p-1}\left\{ 2\sum_{m=\frac{1}{2}}^{\frac{N-1}{2}}\cos\left[\frac{2\pi m\left(q-p+t\right)}{N+1}\right]\right\} ,\label{eq:HardZeroSUM}
\end{equation}
where we used evenness of cosines, to let $m$ run from $1/2$, and
switched the order of the sums. We now prove that the sum inside braces
is $\left(-1\right)^{t}$ . Let $\theta=\frac{2\pi\left(q-p+t\right)}{N+1}$,
$n=m-1/2$ and $N'=\frac{N}{2}-1$ to rewrite the sum 
\[
2\sum_{m=\frac{1}{2}}^{\frac{N-1}{2}}\cos\left[\frac{2\pi m\left(q-p+t\right)}{N+1}\right]\equiv2\sum_{n=0}^{N'}\cos\left[\left(n+\frac{1}{2}\right)\theta\right]
\]
but $\cos\left[\left(n+\frac{1}{2}\right)\theta\right]=\cos\left(n\theta\right)\cos\left(\frac{\theta}{2}\right)-\sin\left(n\theta\right)\sin\left(\frac{\theta}{2}\right)$
and \cite{Wolfram} 
\begin{eqnarray}
\sum_{n=0}^{N'}\cos\left(n\theta\right) & = & \frac{\cos\left(\frac{N'\theta}{2}\right)\sin\left[\frac{\theta}{2}\left(N'+1\right)\right]}{\sin\left(\theta/2\right)}\label{eq:SumCosines}\\
\sum_{n=0}^{N'}\sin\left(n\theta\right) & = & \frac{\sin\left(\frac{N'\theta}{2}\right)\sin\left[\frac{\theta}{2}\left(N'+1\right)\right]}{\sin\left(\theta/2\right)}.\label{eq:SumSines}
\end{eqnarray}
Therefore $\sum_{n=1}^{N'}\cos\left(n+\frac{1}{2}\right)\theta=\sum_{n=1}^{N'}\left\{ \cos\left(n\theta\right)\cos\left(\frac{\theta}{2}\right)-\sin\left(n\theta\right)\sin\left(\frac{\theta}{2}\right)\right\} $
gives
\begin{eqnarray*}
2\sum_{n=0}^{N'}\cos\left(n+\frac{1}{2}\right)\theta & = & 2\left\{ \cos\left(\frac{\theta}{2}\right)\frac{\cos\left(\frac{N'\theta}{2}\right)\sin\left[\frac{\theta}{2}\left(N'+1\right)\right]}{\sin\left(\theta/2\right)}-\sin\left(\frac{\theta}{2}\right)\frac{\sin\left(\frac{N'\theta}{2}\right)\sin\left[\frac{\theta}{2}\left(N'+1\right)\right]}{\sin\left(\theta/2\right)}\right\} \\
 & = & 2\frac{\sin\left[\frac{\theta}{2}\left(N'+1\right)\right]}{\sin\left(\theta/2\right)}\left\{ \cos\left(\frac{N'\theta}{2}\right)\cos\left(\frac{\theta}{2}\right)-\sin\left(\frac{N'\theta}{2}\right)\sin\left(\frac{\theta}{2}\right)\right\} \\
 & = & 2\frac{\sin\left[\frac{\theta}{2}\left(N'+1\right)\right]}{\sin\left(\theta/2\right)}\left\{ \cos\left(\frac{\theta}{2}\left(N'+1\right)\right)\right\} =\frac{\sin\left[\left(N'+1\right)\theta\right]}{\sin\left(\theta/2\right)}=\frac{\sin\left[N\theta/2\right]}{\sin\left(\theta/2\right)}.
\end{eqnarray*}
However $\sin\left[N\theta/2\right]=\sin\left[\frac{\left(N+1\right)\theta}{2}-\frac{\theta}{2}\right]=\sin\left(\frac{\left(N+1\right)\theta}{2}\right)\cos\frac{\theta}{2}-\cos\left(\frac{\left(N+1\right)\theta}{2}\right)\sin\frac{\theta}{2}$.
But $\sin\left(\frac{\left(N+1\right)\theta}{2}\right)=0$, leaving
us with
\begin{equation}
\frac{\sin\left[N\theta/2\right]}{\sin\left(\theta/2\right)}=\frac{\sin\left[\frac{\pi N\left(q-p+t\right)}{N+1}\right]}{\sin\left(\frac{\pi\left(q-p+t\right)}{N+1}\right)}=-\cos\left(\frac{\left(N+1\right)\theta}{2}\right)=-\cos\left(\pi\left(q-p+t\right)\right)=-\left(-1\right)^{q-p+t}.\label{eq:SineRatio}
\end{equation}

Putting this back into the sum (Eq. \eqref{eq:HardZeroSUM})

\begin{eqnarray*}
G\left(r,s\right) & = & \frac{\left(-1\right)^{p+q}}{\left(N+1\right)}\mbox{ }\sum_{t=0}^{2p-1}\left\{ \left(-1\right)^{q-p+t}\right\} \\
 & = & \frac{\left(-1\right)^{2q}}{\left(N+1\right)}\mbox{ }\sum_{t=0}^{2p-1}\left(-1\right)^{t}=\frac{1}{\left(N+1\right)}\mbox{ }\sum_{t=0}^{2p-1}\left(-1\right)^{t}\quad;
\end{eqnarray*}
zero comes out because we are summing alternating $+1$'s and $-1$'s
an even number of times. This completes the proof of the harder zeros.
Note that we used $q>p$. For example if $q=p$, then $\cos\pi\left(q-p+t\right)$
would be $1$ for $t=0$ and the sum would give a $2p-1$ on that
term alone.

Recall $\frac{r+1}{2}=q$ and $\frac{s}{2}=p$ with integers $p$
and $q$ satisfying $0<p<q\le N/2$; for this choice
\[
G\left(r,s\right)=0\quad.
\]

\subsection{Nonzero entries: $\pm1$'s in the $\mathbf{G}$}

It remains to show that when $r$ is even and $s$ is odd, $G\left(r,s\right)$
is $\pm1$ as shown in Eq. \eqref{eq:The-Green's-function}. Let $r=2q$
and $s=2p-1$ with $p\le q$ (note that we allow for equality as well).
Using Eq. \eqref{eq:Gmassaged} and previous techniques we have
\begin{eqnarray*}
G\left(r,s\right) & = & -\frac{ie^{i\left(2\left(p+q\right)-1\right)\pi/2}}{2\left(N+1\right)}\sum_{m=-\frac{N-1}{2}}^{\frac{N-1}{2}}e^{i2\omega m\left(p+q-1\right)}\frac{\left[1-e^{-2i\left(2q\right)\omega m}\right]\left[1+e^{-2i\left(2p-1\right)\omega m}\right]}{1-e^{-2i\omega m}}\\
 & = & -\frac{ie^{i\left(2\left(p+q\right)-1\right)\pi/2}}{2\left(N+1\right)}\sum_{m=-\frac{N-1}{2}}^{\frac{N-1}{2}}\left(e^{2i\omega m\left(p+q-1\right)}+e^{-2i\omega m\left(p-q\right)}\right)\\
 & \times & \left\{ 1+e^{-2i\omega m}+e^{-4i\omega m}+\cdots+e^{-2i\left(2q-1\right)\omega m}\right\} .
\end{eqnarray*}
Once again we multiply the parenthesis into the braces to get (using
$ie^{i\left(2\left(p+q\right)-1\right)\pi/2}=e^{i\left(p+q\right)\pi}$)
\begin{eqnarray}
G\left(r,s\right) & = & -\frac{e^{i\left(p+q\right)\pi}}{2\left(N+1\right)}\sum_{m=-\frac{N-1}{2}}^{\frac{N-1}{2}}\left\{ \left[e^{2i\omega m\left(p+q-1\right)}+e^{2i\omega m\left(p+q-2\right)}+\cdots+e^{2i\omega m\left(p-q\right)}\right]\right.\label{eq:pm_expanded}\\
 & + & \left.\left[e^{-2i\omega m\left(p-q\right)}+e^{-2i\omega m\left(p-q+1\right)}+\cdots+e^{-2i\omega m\left(p+q-1\right)}\right]\right\} \nonumber 
\end{eqnarray}

Comment: Eq. \eqref{eq:pm_expanded} looks very similar to Eq. \eqref{eq:HardZeros_CosineForm};
however, it has a key difference. Since $q\ge p$, in either one of
the brackets there will be a term with exponent zero. For example,
if one looks at the first brackets the first term is $e^{2i\omega m\left(p+q-1\right)}$
, which clearly has a positive exponent; however, the last term $e^{2i\omega m\left(p-q\right)}$
has either zero or negative exponent. If it is negative, then a term
preceding it must have had zero exponent. Therefore, the sum for some
choice of $r$ and $s$ can look like
\begin{eqnarray}
G_{example}\left(r,s\right) & = & -\frac{e^{i\left(p+q\right)\pi}}{2\left(N+1\right)}\sum_{m=-\frac{N-1}{2}}^{\frac{N-1}{2}}\left\{ \left[e^{2i\omega m\left(p+q-1\right)}+\cdots+e^{2i\omega m}+1+e^{-2i\omega m}+\cdots+e^{2i\omega m\left(p-q\right)}\right]\right.\label{eq:pm_expanded-1}\\
 & + & \left.\left[e^{-2i\omega m\left(p-q\right)}+e^{2i\omega m}+1+e^{-2i\omega m}+\cdots+e^{-2i\omega m\left(p+q-1\right)}\right]\right\} \quad.\nonumber 
\end{eqnarray}
We can pair the terms to the left (right) of the $1$ in the first
bracket with those to the right (left) of the $1$ in the right bracket
to get the cosines as before. It is clear that the sum over $2$ contributes
a $2\left(N-1\right)$. We now show that the sum over the cosines
contributes a $4$, which together makes $2\left(N+1\right)$ and
cancels the denominator in the pre-factor.

For any $p$ and $q$, we can find a $t_{0}=q-p\ge0$ that makes the
exponent zero. In the first bracket, there are $q-p$ terms to its
left and there are $2q-\left(q-p+1\right)=p+q-1$ terms to its right
(for a total of $2q$ terms). We can pair the terms to its left with
the corresponding terms in the second bracket (now to the right of
the $1$) to get cosines and similarly pair terms to its right to
get cosines. Then, we can break the sum in the foregoing equation
to the sum over cosines obtained from terms to the left of $t_{0}$
in the first bracket, the sum over terms to its right and add a $2$
for the term itself. Namely
\begin{eqnarray}
G\left(r,s\right) & = & -\frac{e^{i\left(p+q\right)\pi}}{2\left(N+1\right)}\sum_{m=-\frac{N-1}{2}}^{\frac{N-1}{2}}\left\{ 2\sum_{t=1}^{q-p}\cos\left[2\omega mt\right]+2+2\sum_{t=1}^{p+q-1}\cos\left[2\omega mt\right]\right\} \label{eq:Grs_sum_pm1}\\
 & = & -\frac{e^{i\left(p+q\right)\pi}}{\left(N+1\right)}\sum_{m=-\frac{N-1}{2}}^{\frac{N-1}{2}}\left\{ \sum_{t=1}^{q-p}\cos\left[2\omega mt\right]+\sum_{t=1}^{p+q-1}\cos\left[2\omega mt\right]+1\right\} ,\nonumber 
\end{eqnarray}
where we cancelled the overall pre-factor of a $1/2$. Let us evaluate
each of the sums separately (switching order of summation, changing
variables as before) 
\begin{eqnarray*}
\sum_{m=-\frac{N-1}{2}}^{\frac{N-1}{2}}\cos\left[2\omega mt\right] & = & 2\sum_{n=0}^{\frac{N}{2}-1}\cos\left[2\omega t\left(n+\frac{1}{2}\right)\right]\\
 & = & 2\cos\left(\omega t\right)\sum_{n=0}^{\frac{N}{2}-1}\cos\left(2\omega tn\right)-2\sin\left(\omega t\right)\sum_{n=0}^{\frac{N}{2}-1}\sin\left(2\omega tn\right)
\end{eqnarray*}

The sum over cosines are evaluated using Eq. \eqref{eq:SumCosines},\ref{eq:SumSines}
\begin{eqnarray*}
\sum_{n=0}^{\frac{N}{2}-1}\cos\left(2\omega tn\right) & = & \frac{\cos\left(\left(\frac{N}{2}-1\right)\omega t\right)\sin\left[\omega t\left(N/2\right)\right]}{\sin\left(\omega t\right)}\\
\sum_{n=0}^{\frac{N}{2}-1}\sin\left(2\omega tn\right) & = & \frac{\sin\left(\left(\frac{N}{2}-1\right)\omega t\right)\sin\left[\omega t\left(N/2\right)\right]}{\sin\left(\omega t\right)}
\end{eqnarray*}
which together give 
\begin{eqnarray*}
\sum_{m=-\frac{N-1}{2}}^{\frac{N-1}{2}}\cos\left[2\omega mt\right] & = & 2\frac{\sin\left[\left(\frac{N}{2}\right)\omega t\right]}{\sin\left(\omega t\right)}\left\{ \cos\left(\left(\frac{N}{2}-1\right)\omega t\right)\cos\left(\omega t\right)-\sin\left(\left(\frac{N}{2}-1\right)\omega t\right)\sin\left(\omega t\right)\right\} \\
 & = & 2\frac{\sin\left[\left(\frac{N}{2}\right)\omega t\right]}{\sin\left(\omega t\right)}\left\{ \cos\left[\left(\frac{N}{2}\right)\omega t\right]\right\} =\frac{\sin\left(N\omega t\right)}{\sin\left(\omega t\right)}=\frac{\sin\left(\frac{\pi Nt}{N+1}\right)}{\sin\left(\frac{\pi t}{N+1}\right)}.
\end{eqnarray*}
We calculated this ratio in Eq. \eqref{eq:SineRatio} so we have 
\[
\frac{\sin\left(\frac{\pi Nt}{N+1}\right)}{\sin\left(\frac{\pi t}{N+1}\right)}=-\left(-1\right)^{t}\quad.
\]
Using this we can evaluate Eq. \eqref{eq:Grs_sum_pm1}
\[
G\left(r,s\right)=-\frac{e^{i\left(p+q\right)\pi}}{\left(N+1\right)}\left\{ -\sum_{t=1}^{q-p}\left(-1\right)^{t}-\sum_{t=1}^{p+q-1}\left(-1\right)^{t}+\sum_{m=-\frac{N-1}{2}}^{\frac{N-1}{2}}1\right\} .
\]
If $q-p$ is even then $q-p=2k$ for some $k$ and $p+q-1=2\left(k+p\right)-1$,
which is odd. Also if $q-p=2k-1$, then $p+q-1=2\left(p+k-1\right)$,
which is even. In either case one of the sums vanishes and the other
evaluates to be $-1$. Therefore,

\begin{eqnarray}
G\left(r,s\right) & = & -\frac{e^{i\left(p+q\right)\pi}}{\left(N+1\right)}\left\{ 1+\sum_{m=-\frac{N-1}{2}}^{\frac{N-1}{2}}1\right\} =-\frac{e^{i\left(p+q\right)\pi}}{\left(N+1\right)}\left(N+1\right)=-e^{i\left(p+q\right)\pi}\label{eq:pm1_pqform}
\end{eqnarray}

So we predict that if $p+q$ is odd then $G\left(r,s\right)=+1$ and
if $p+q$ is even then $G\left(r,s\right)=-1$. What does this mean
for $r$ and $s$? Let us cover all of the cases one by one. Recall
that $r$ is even and $s$ is odd and $q+p=\frac{r}{2}+\frac{s+1}{2}$
\begin{itemize}
\item $ $$p+q$ is even, $G\left(r,s\right)=-1$, and $q$ is even. This
means, $p$ is even. These implies that $r$ and $s+1$ are multiples
of $4$. Looking at $\mathbf{G}$, we see that these entries indeed
are $-1$.
\item $p+q$ is even, $G\left(r,s\right)=-1$, and $q$ is odd. This means
$p$ is odd. These imply that $r$ and $s+1$ are multiples of $2$
but not $4$. Looking at $\mathbf{G}$, we see that the rest of the
entries that are $-1$ have been covered.
\item $p+q$ is odd, $G\left(r,s\right)=+1$, and $q$ is even. This means
$p$ is odd. These imply that $r$ is a multiple of $4$ but $s+1$
is not (though of course even). This covers some of the $+1$'s in
$\mathbf{G}$.
\item Lastly, $p+q$ is odd, $G\left(r,s\right)=+1$, and $q$ is odd. This
means $p$ is even. These imply that $r$ is not a multiple of $4$
(though of course even), yet $s+1$ is a multiple of $4$. These cover
the rest of $+1$'s seen $\mathbf{G}$,.
\end{itemize}
The final result Eq. \eqref{eq:pm1_pqform} can be expressed in terms
of $r$ and $s$ as 
\begin{eqnarray}
G\left(r,s\right) & = & -e^{i\pi\left(\frac{r+s+1}{2}\right)}=e^{i\pi\left(\frac{r+s-1}{2}\right)}\quad.\label{eq:G_pm1_Final}
\end{eqnarray}

This completes our proof. \\

\begin{rem}
All the equations above for $G\left(r,s\right)$ were checked numerically.\end{rem}

\end{document}